%% file: main.tex
\documentclass[runningheads]{llncs} 
\usepackage{amssymb}
\usepackage{xspace}
\usepackage{amsmath}
\usepackage[dvipsnames]{xcolor}
\usepackage{graphicx}
\usepackage{multicol}
\usepackage{paralist}
\usepackage{comment}
\usepackage{enumitem}
\usepackage{wrapfig}
\usepackage{ifthen}
\newboolean{long}
\setboolean{long}{false}
\usepackage{bbm}
\usepackage{pgf,pgfarrows,pgfnodes}
\usepackage{tikz}
\usetikzlibrary{arrows,automata}
\usepackage{url}
\usepackage{vwcol}
\usepackage{algorithm,algorithmicx,algpseudocode}
\usepackage[normalem]{ulem}

\input{macros}

\begin{document}
\sloppypar

\title{Logic-based Specification and Verification of Homogeneous Dynamic Multi-agent Systems\\
  {\footnotesize (To appear in Journal of \textsc{AAMAS}, revised version)}
}

\titlerunning{Logic-based Verification of \oursys}        

\author{Riccardo {De Masellis}         \and
        Valentin Goranko 
}

\authorrunning{R. De Masellis, V. Goranko} 

\author{Riccardo {De Masellis}\inst{1} \and Valentin Goranko\inst{2}}

\institute{Stockholm University,
\email{riccardo.demasellis@philosophy.su.se} \\
\and
Stockholm University and University of Johannesburg (visiting professorship),
\email{valentin.goranko@philosophy.su.se}
}

\date{Received: date / Accepted: date}

\maketitle

\begin{abstract}  
\input{abstract.tex}

\end{abstract}

\input{introduction.tex}

\input{framework.tex}

\input{logic.tex}
\input{modelChecking.tex}

\input{complexity.tex}

\input{concluding.tex}



\bibliographystyle{splncs04}     
\bibliography{VG-MAS}   

 \end{document}

%% file: macros.tex
\usepackage[colorinlistoftodos,textwidth=60pt]{todonotes}

\newcommand{\nat}{\ensuremath{\mathbb{N}}}

\newcommand{\coop}[2][]{\langle\!\langle{#2}\rangle\!\rangle_{_{\!\mathit{#1}}}\,}

\newcommand{\atlx}{\mathord \mathsf{X}\, }
\newcommand{\atlf}{\mathord \mathsf{F}\, }
\newcommand{\atlg}{\mathord \mathsf{G}\, }
\newcommand{\atlu}{\, \mathsf{U} \, }

\newcommand{\ifff}{\leftrightarrow}
\newcommand{\cut}[1]{}
\newcommand{\vcut}[1]{ }

\newcommand{\set}[1]{\{{#1}\}}

\newcommand{\bbN}{\ensuremath{\mathbb{N}}}

\newcommand{\defstyle}{\textbf}

\newcommand{\ra}{\ensuremath{\rightarrow}}

\newcommand{\ol}[1]{\overline{#1}}

\newcommand{\tuple}[1]{\ensuremath{\langle #1 \rangle}}

\newcommand{\oursys}{\textsc{hdmas}\xspace}

\newcommand{\agset}{\ensuremath{\mathit{Ag}}\xspace}
\newcommand{\ag}{\ensuremath{\mathit{ag}}}
\newcommand{\coal}{\ensuremath{\boldsymbol{C}}}
\newcommand{\opp}{\ensuremath{\boldsymbol{N}}}

\newcommand{\actset}{\ensuremath{\mathit{Act}}}
\newcommand{\act}{\ensuremath{\mathit{act}}}
\newcommand{\actav}{\ensuremath{\mathit{d}}}
\newcommand{\States}{\ensuremath{\mathit{S}}\xspace}
\newcommand{\sts}{\ensuremath{\mathit{s}\xspace}}

\newcommand{\atprop}{\ensuremath{\mathit{AP}}}

\newcommand{\guardset}{\ensuremath{G}}
\newcommand{\guard}{\ensuremath{\mathit{g}}}

\newcommand{\dmas}{\ensuremath{\mathcal{M}}}

\newcommand{\actvar}{\ensuremath{x}}
\newcommand{\actvarset}{\ensuremath{\mathit{X}}}
\newcommand{\varsetf}[1]{\ensuremath{\mathit{Var(#1)}}}

\newcommand{\labf}{\ensuremath{\lambda}}

\newcommand{\stratg}{\ensuremath{\sigma}}
\newcommand{\stratgn}[1]{\ensuremath{\stratg_{#1}}}

\newcommand{\outf}{\ensuremath{\mathit{out}}}

\newcommand{\actass}{\ensuremath{\eta}}
\newcommand{\actassset}{\ensuremath{H}}

\newcommand{\actassr}[1]{\ensuremath{\actass|_{#1}}}
\newcommand{\actasssetr}[1]{\ensuremath{\actassset|_{#1}}}
\newcommand{\actasssetn}[1]{\ensuremath{\actassset|^{#1}}}

\newcommand{\actasssum}[1]{\ensuremath{\mathsf{sum}(#1)}}

\newcommand{\agass}{\ensuremath{\theta}}
\newcommand{\agassow}[2]{\ensuremath{\agass[#1 := #2]}}

\newcommand{\agvarset}{\ensuremath{Y}}
\newcommand{\agvar}{\ensuremath{y}}

\newcommand{\actmap}{\ensuremath{\mathsf{\mu}}}
\newcommand{\noop}{\ensuremath{\varepsilon}}

\newcommand{\actprof}[1]{\ensuremath{\pi_{#1}}}
\newcommand{\actprofr}[2]{\ensuremath{\actprof{#1}|_{#2}}}
\newcommand{\actprofset}[1]{\ensuremath{\Pi_{#1}}}
\newcommand{\actprofsetr}[2]{\ensuremath{\actprofset{#1}|_{#2}}}

\newcommand{\abst}{\ensuremath{\alpha}}

\newcommand{\termset}{\ensuremath{T}}
\newcommand{\term}{\ensuremath{t}}

\newcommand{\pth}{\ensuremath{w}}

\newcommand{\transrel}{\ensuremath{\delta}}
\newcommand{\exptransf}{\ensuremath{\Delta}}

\newcommand{\dom}{\ensuremath{\mathit{dom}}}

\newcommand{\stexten}[2]{[\![ #1]\!]_{\dmas}^{#2}}

\newcommand{\algostyle}[1]{\textsc{#1}}

\newcommand{\astratg}{\ensuremath{\rho}}

\newcommand{\grd}[2]{\guard^{#1}_{#2}}
\newcommand{\prf}{\ensuremath{\mathit{prf}}}

\newcommand{\pfix}{\ensuremath{\mathsf{pfix\xspace}}}

\newcommand{\emptystring}{\ensuremath{\epsilon}}

\newcommand{\lang}{\ensuremath{\mathcal{L}_{\oursys}}\xspace}
\newcommand{\fineq}{\ensuremath{\equiv_{fin}}}

\newcommand{\longversion}[1]{}

\renewcommand{\actass}{\ensuremath{\mathbf{act}\xspace}}
\newcommand{\langb}{\ensuremath{\lang^{0}}\xspace}

\newcommand{\Q}{\ensuremath{\mathsf{Q}\xspace}}

\newcommand{\cmon}{\textbf{(C-mon)}}
\newcommand{\nmon}{\textbf{(N-mon)}}

\newcommand{\parset}{\ensuremath{Z}}
\newcommand{\parvar}{\ensuremath{z}}

\newcommand{\pra}{\ensuremath{\mathsf{PrA}\xspace}}

\renewcommand{\fineq}{\ensuremath{\equiv_{\mathsf{fin}}}}

\renewcommand{\langb}{\ensuremath{\lang^{\mathsf{NF}}}\xspace}

\renewcommand{\prf}{\ensuremath{\mathsf{PrF}}}

\newcommand{\hdmas}{\oursys}

\newcommand{\nd}{\ensuremath{\mathbf{nd}\xspace}}

\renewcommand{\actprofset}[1]{\ensuremath{\mathsf{P}_{#1}}}
\renewcommand{\actprof}[1]{\ensuremath{\mathsf{p}_{#1}}}

\renewcommand{\pth}{\ensuremath{\pi}\xspace}

\newcommand{\setagt}{\ensuremath{A}\xspace}
\newcommand{\stu}{\ensuremath{u}\xspace}

\newcommand{\statw}{\ensuremath{w}\xspace}

\newcommand{\astate}{\sts}
\newcommand{\fo}{\varphi}
\newcommand{\fob}{\psi}
\newcommand{\aset}{Z}
\newcommand{\asetbis}{W}
\newcommand{\avarprop}{p}
\algblockdefx{Cases}{EndCases}%
   [1]{\textbf{case} #1 \textbf{of}}%
   {\textbf{end case}}
\algcblockx[Cases]{Cases}{Case}{EndCases}%
   [1]{#1:\enspace}%
   {\textbf{end case}}

   \newcommand{\Qstates}{\ensuremath{Q}}

   \newcommand{\pushf}{\algostyle{push}\xspace}
   \newcommand{\nff}{\algostyle{nf}\xspace}
    \newcommand{\pqe}{\algostyle{pqe}\xspace}


%% file: abstract.tex
We develop a logic-based framework for formal specification and algorithmic verification of \emph{homogeneous} and \emph{dynamic} concurrent multi-agent transition systems (HDMAS). 
Homogeneity means that all agents have the same available actions at any given state and the actions have the same effects regardless of which agents perform them. The state transitions are therefore determined only by the vector of numbers of agents performing each action and are specified symbolically, by means of conditions on these numbers definable in Presburger arithmetic.
The agents are divided into \emph{controllable} (by the system supervisor/controller) and \emph{uncontrollable}, representing the environment or adversary. Dynamicity means that the numbers of controllable and uncontrollable agents may vary throughout the system evolution, possibly at every transition.

As a language for formal specification we use a suitably extended version of Alternating-time Temporal Logic (ATL), where one can specify properties of the type ``a coalition of (at least) $n$ controllable agents can ensure against (at most) $m$ uncontrollable agents that any possible evolution of the system satisfies a given objective $\gamma$'', where $\gamma$ is specified again as a formula of that language and each of $n$ and $m$ is either a fixed number or a variable that can be quantified over.

We provide formal semantics to our logic \lang\ and define normal form of its formulae. 
We then prove that every formula in \lang\  is equivalent in the finite to one in a normal form and develop an algorithm for global model checking of formulae in normal form in finite HDMAS models, which invokes model checking truth of Presburger formulae. We establish worst case complexity estimates for the model checking algorithm and illustrate it on a running example.

%% file: introduction.tex
\section{Introduction}

\textbf{The framework.}
We consider discrete concurrent multi-agent transition systems, i.e. multi-agent systems (MAS) in which the transitions take place in a discrete succession of steps, as a result of a simultaneous (or, at least mutually independent) actions performed by all agents. Such MAS are typically modelled as \emph{concurrent game models} (cf \cite{AHK-02} or \cite{BGJ15}).

Here we focus on a special type of concurrent MAS, which are  \emph{homogeneous} and \emph{dynamic}, in a sense explained below. 

The  \emph{homogeneity} means that all agents are essentially indistinguishable from each other, as their possible behaviours are determined by the same protocol. In particular, they have the same available actions at each state and the effect of these actions depends not on \emph{which} agents perform them, but only on \emph{how many} agents perform each action.  
Thus, the transitions in such systems are determined \emph{not by the specific action profiles}, but only \emph{by the vector of numbers of agents that perform each of the possible actions in these action profiles}. The latter can be regarded as an abstraction of the action profile. The transitions are specified symbolically, by means of conditions on these vectors, definable in Presburger arithmetic.

Typical examples of such homogeneous systems include: 

\begin{itemize}
\item \emph{voting procedures}
 where the outcome only depends on how many agents vote  for each possible alternative, but not who votes for what. These also involve voting procedures where anonymity is required and the identity of agents should not be inferred by observing the system's evolution~\cite{PedersenD13,DBLP:conf/voteid/JamrogaKK18};

\item \emph{sensor networks} of a type where protocols only depend on how many sensors send any given signal~\cite{VinyalsEtAl-2011};

\item \emph{computer network servers}, the functioning of which only depends on how many currently connected users are performing any given action (e.g. uploading or downloading data, sending printing jobs, communicating over common channels, etc);  

\item \emph{markets}, the dynamics of which only depends on how many agents are selling and how many are buying any given stock (assuming the transactions are per unit) but not exactly who does what.
\end{itemize}

The  \emph{dynamicity} of the systems that we consider means that the set (hence, the number) of agents being present (or, just acting) in the system may vary throughout the system evolution, possibly at every transition from a state to a state. All examples listed above naturally have that dynamic feature. There are different ways to interpret such dynamicity. In the extreme version, agents literally appear and disappear from the system, e.g. users joining and leaving an open network. A less radical interpretation is where the agents are in the system all the time but may become active and inactive from time to time, e.g. voters, or members of a committee, may abstain from voting in one election or decision making round, and then become active again in the next one. A more refined version is where at every state of the system performance each agent decides to act (i.e. take one of the available actions) or pass/idle, formally by performing the `pass/idle' action. Technically, all these interpretations seem to be reducible to the latter one.
However, the way we model the dynamicity here is by assuming that there is an unbounded, and  possibly infinite set of `potentially existing' agents, but that only finitely many of them are `actually existing/present' at each stage of the evolution of the system. Therefore, at each transition round, only finitely many currently existing agents can possibly perform an action, and each of these may also choose not to perform any action (i.e.,  remain inactive in that round). However, the currently inactive (or, `non-existing') agents do not have any individual influence on the transitions. Thus, the number of currently active agents, who determine the next transition, can change from any instant to the next one, while always remaining finite. We note, however, the difference between dynamic systems, in the sense described above, and simply \emph{parametric} systems, where the number of agents is taken as a parameter but remains fixed during the whole evolution of the system. In that sense, the present study applies both to parametric and truly dynamic systems.

In this work we develop a logic-based framework for formal specification and algorithmic verification of the behaviour of homogeneous dynamic multi-agent systems (\oursys) of the type described above. We focus, in particular, on scenarios where the agents are divided into \emph{controllable} (by the system supervisor or controller) and  \emph{uncontrollable}, representing the environment or an adversary. Both numbers, of controllable and uncontrollable agents, may be fixed or varying throughout the system evolution, possibly at every transition. 
The controllable agents are assumed to act according to a joint strategy prescribed by the supervisor/controller, with the objective to ensure the desired behaviour of the system (e.g. reaching an outcome in the voting procedure, or keeping the demand and supply of a given stock within desired bounds, or ensuring that the server will not be deadlocked by a malicious attack of adversary users, etc). 

As a logical language for formal specification we introduce a suitably extended version,  \lang,  of the Alternating time temporal logic ATL (\cite{AHK-02}). In \lang\  one can specify properties of the type {``\textit{A team of (at least) $n$ controllable agents can ensure, against at most $m$ active uncontrollable agents, that any possible evolution of the system satisfies a given objective $\varphi$}''},
where the objective $\varphi$ is specified again as a formula of that language, and each of $n$ and $m$ is either a fixed number, a  parameter, or a {variable} that can be quantified over. 

To summarise the comparison: in the standard concurrent game models of MAS agents are explicitly distinguished and in the logic ATL they are explicitly referred to by their names (individually, or in coalitions). In the HDMAS framework developed here, the only distinction between the agents is whether they are controllable or not, and in the language both are referred to only by numbers.

Here is an indicative, yet generic scenario, where our framework is readily applicable for both modelling and verification.

\textit{
  A military fortress has $k$ protected points of entry: $A_1, A_2 … A_k$, with $k > 2$. The commander of the fortress has $C$ soldiers, hereafter called `defenders’, that can be deployed to protect these points of entry against an invading army. For each $A_i$, a number $c_i$ of defenders, with $m_i \leq c_i \leq M_i$, can be deployed against $n_i$ `invaders'. If $c_i = M_i$, then the defenders successfully protect $A_i$ against any number of invaders; if $c_i \not = M_i$, then entry point $A_i$ is lost when $n_i > m_i$.
  Moreover, both the defender and the invading commander may receive reinforcements and re-deploy their soldiers among the entry points once a day (say, at noon), whereas the attacks can only take place at night.
  However, neither of them can observe the precise distribution of the soldiers of the other party, but they can observe which points of entry are currently ``outpowered’’ by not being sufficiently protected by defenders. It is also known that the enemy must outpower more than 2 points of entry at the same time in order to successfully invade the fortress.
  }

The framework \oursys that we develop here will enable modelling the scenario above as well as specifying and algorithmically verifying claims of the kind: ``\textit{The fortress commander has a strategy to protect the fortress for at least $d$ days, for a given $d$ (or, forever) with $C$ defenders against at most $V$ (or, against any number of) invaders}''.

\medskip
\textbf{Structure and content of the paper.} 
In Section \ref{sec:framework} we introduce the \oursys framework,  provide a running example, and prove some technical results needed to introduce \emph{counting abstractions} of joint actions and strategy profiles. Using these counting abstractions, in Section \ref{sec:logic} we provide formal semantics in \oursys models for the logic \lang\ which we introduce there. We then define normal form of formulae of  \lang\ and the fragment \langb, consisting of formulae in normal form. The key technical result obtained in that section is that every formula in \lang\  is equivalent on finite models to one in \langb. In Section \ref{sec:MC} we develop an algorithm for global model checking of formulae in \langb in finite \oursys models, which invokes model checking truth of their respective translations into Presburger formulae, and illustrate that algorithm on running examples. In Section \ref{sec:complexity} we establish some refined  complexity estimates for the model checking algorithm, using recent complexity results obtained in~\cite{DBLP:conf/csl/Haase14} for fragments of Presburger arithmetic. We end with some concluding remarks on extensions and possible applications of our work in Section \ref{sec:concluding}. 

\medskip
\textbf{Related work.}  
{A more closely related framework to ours is Open Multi-Agent Systems (OMAS)~\cite{KouvarosLPP19}. \oursys shares with it the characteristic 'dynamic' feature of agents, which can therefore leave and join the system at runtime. However \oursys differs from OMAS in several essential aspects. First, although any finite number of agents can perform actions at each step, the evolution of OMAS depends only on the projection of those on the set of actions or, in other words, whether any action is performed by at least one agent. 
Thus, \oursys makes use of the full expressivity of Presburger arithmetic. Next, the verification formalism of OMAS is a temporal epistemic logic with (universally quantified) indices spanning over agents, while ours includes strategic operators. Lastly, decidability of model-checking Open Multi-agent Systems is obtained by restricting the semantics of the models and by using cutoff techniques whereas we ultimately invoke model-checking truth of Preseburger formulas.

We are aware of other threads of, more or less essentially, related work, however none of them considers formal models and verification methods for the type of homogeneous and dynamic multi-agent scenarios studied here. Therefore, we only mention them briefly as in all frameworks mentioned below, the number of agents is fixed along system executions, possibly as a parameter and the formal specification languages do not explicitly allow quantification over the number of agents.}

-- Counting abstraction for verification of parametric systems has been studied in~\cite{GermanS92} and~\cite{BloemJKKRVW16}, where techniques based on Petri nets or Vector Addition Systems with States (VASS) are used to obtain decidability of model checking.

-- The work in~\cite{RaskinSB05} is closer to ours, as strategic reasoning is considered but only for a restricted set of properties such as reachability, coverability and deadlock avoidance. Also, assumptions on the system evolutions are made and, in particular, monotonicity with respect to a well-quasi-ordering.

-- In~\cite{KouvarosL16} temporal epistemic properties of parametric interpreted systems are checked irrespective of  the number of agents by using cutoff techniques.

-- Modular Interpreted Systems~\cite{Jamroga2007MIS} is a MAS framework where
a decoupling between local agents and global system description is achieved, thus possibly amenable to model dynamical MAS frameworks.

-- Homogeneous MAS with transitions determined by the number of acting agents have been introduced in~\cite{PedersenD13}.

-- Population protocols~\cite{AngluinADFP04} are parametric systems of homogeneous agents, and decidability of model checking against probabilistic linear-time specification is studied in~\cite{EsparzaGLM16}.

-- In~\cite{DeGiacomoVFAL18}, instead of verifying MAS with unknown number of agents, the authors propose a technique to find the minimal number of agents which, once deployed and suitably orchestrated, can carry out a manufacturing task.

-- Lastly, as noted above, our logic of specification builds on the 
Alternating time temporal logic ATL (\cite{AHK-02}) and extends the model checking algorithm for ATL to \hdmas. 


%% file: framework.tex
\section{Preliminaries and modelling framework}
\label{sec:framework}

We start by introducing the basic ingredients of our framework.
We assume a hereafter fixed (finite, or possibly countably infinite) \defstyle{universe of potential agents} $\agset = \{\ag_1, \ag_2, \ldots\}$, but only finitely many of them will be assumed currently present, or `currently existing', at any time instant or stage of the evolution of the system. 
{Alternatively, the universe of agents can be assumed always finite but unbounded.} 

Next, we consider a finite set of \defstyle{action names} $\actset=\set{\act_1, \ldots, \act_n}$. 
We extend this set  with a specific `idle' action $\noop$ and define 
$\actset^+= \actset \cup \set{\noop}$. 
We also fix a set of distinct variables $\actvarset = \set{\actvar_1, \ldots, \actvar_n}$ extended to $\actvarset^+= \actvarset \cup \set{x_{\noop}}$, called \defstyle{action counters}, associated to $\actset$ and $\actset^+$ respectively. Formally, we relate these by a mapping $\actmap: \actset^+ \ra \actvarset^+$ such that for each $i \in \set{1, \ldots, n}$, $\actmap(\act_i)=\actvar_i$ and $\actmap(\noop)=\actvar_{\noop}$.
Hereafter, $\actset$, $\actset^+$, $\actvarset$, $\actvarset^+$, and $\actmap$ are assumed fixed, as above. 

An \defstyle{action profile} over a given set of actions $\actset' \subseteq \actset^+$ is defined as a function $\actprof{}: \agset \ra \actset'$, 
assigning an action from $\actset'$ to each agent in $\agset$. 
More generally, for any subset of agents $A \subseteq \agset$, a \defstyle{joint action of $A$} over a set of actions $\actset' \subseteq \actset^+$ 
is a function $\actprof{A}$
assigning an action from $\actset'$ to each agent in $A$.

Given a function $f$, we will write: $\dom(f)$ for the domain of $f$;  $f|_{Z}$ for the restriction of $f$ to a domain $Z \subseteq dom(f)$;  
and $f[Z]$ for the image of $Z$ under $f$. For technical purposes, we also consider a (unique) function $f_{\emptyset}$ with an empty domain.

{To express relevant conditions on the number of agents performing actions in $\actvarset$, we make use of Presburger arithmetic (the first-order theory of natural numbers with addition and $=$). This is a fairly expressive, yet decidable theory, which makes it very natural and suitable for many computational tasks related to verification of various discrete infinite-state systems (see e.g. ~\cite{DBLP:journals/siglog/Haase18} for an introduction.)}

\begin{definition}[Guards]
  A \defstyle{(transition) guard} $\guard$ is an open (quantifier-free)\footnote{The restriction to quantifier-free guards is only partly essential for the technical results,  given the quantifier elimination property of Presburger arithmetic. We make that restriction mainly to keep the presentation simpler.} formula  of 
  Presburger arithmetic \pra\ with predicates $=$ and $<$ over variables from the set of action counters $\actvarset$.
    We denote by $\guardset$ the set of all guards, 
    {by $\varsetf{g}$ the set of variables occurring in a guard $g \in \guardset$}, and we use the following standard abbreviations in Presburger formulas: 
    {$n := 1+ \ldots +1$ ($n$ times 1) and $n \actvar := \actvar + \ldots + \actvar$ ($n$ times $\actvar$) for any $n \in \nat$ and $\actvar \in \actvarset^+$}.
\end{definition}

\begin{definition} \label{def:act-dist}
  An \defstyle{action distribution} is any function $\actass : X' \ra \nat$, where 
$X' \subseteq  \actvarset^+$. The domain $X'$ 
is denoted, as usual, by $\dom(\actass)$. 
Intuitively, an action distribution assigns for every action $\act$, through the value 
of the action counter $\actmap(\act)$, the number of agents who are assigned the action $\act$.

Given an action distribution $\actass$ we define:
\begin{itemize}

\item  $\actass \models \guard$, for a given guard $\guard$, if $\actass$ satisfies $\guard$  with the expected standard semantics of \pra, 
{namely: \\ 
$\actass \models \actvar_1 = \actvar_2$ if $\actmap(\actvar_1) = \actmap(\actvar_2)$ and $\actass \models \actvar_1 < \actvar_2$ if $\actmap(\actvar_1) = \actmap(\actvar_2)$};

\item  $\actasssum{\actass} := \sum_{\actvar \in \dom(\actass)} \actass(\actvar)$;

\item $\actasssetn{m} := \set{\actass \mid \actasssum{\actass} = m}$ 
{is the set of action distributions where exactly $m$ agents perform actions};

  \item $\actassset := \bigcup_{m \in \nat} \actasssetn{m}$ 
  {is the set of all action distributions}.
  
 \end{itemize}

We also define the mapping $\oplus : \actassset \times \actassset \dashrightarrow \actassset$,  
  which, given two action distributions $\actass_{1}$ and $\actass_{2}$, is defined if $\dom(\actass_{1}) = \dom(\actass_{2}) := Z$ and
 returns a new action distribution, $\actass_{1} \oplus \actass_{2}$, with domain $Z$, defined component-wise as the sum of $\actass_{1}$ and $\actass_{2}$, i.e. 
 $\actass_{1} \oplus \actass_{2} (z) = \actass_{1}(z) + \actass_{2}(z)$ for each $z \in Z$.
\end{definition}

\begin{remark} \label{rmk:noop-assignments}
Note that guards are defined over the set of variables $\actvarset$, while the domain of action distributions can also include $\actvar_{\noop}$. 
It follows that, for any action distribution $\actass$, the value $\actass(\actvar_{\noop})$ does not have any influence on the satisfiability of a guard. More generally, for every $\actass \in \actassset$ and $\guard \in \guardset$ we have $\actass \models \guard$ iff $\actassr{\varsetf{\guard}} \models \guard$. 
  \end{remark}

We now relate action profiles with action distributions. Every action profile is associated with the action distribution that counts, for each action, the number of agents performing it. In that sense, action distributions are \emph{counting abstractions} for action profiles. The formal definition follows, where we denote the set of all action profiles over $\actset$ by $\actprofset{}$
and define the inverse of an action profile $\actprof{}$ as the function 
$\actprof{}^{-1} : \actset \ra \wp(\agset)$ such that $\actprof{}^{-1}(\act) = \set{\ag \in \agset \mid \actprof{}(\ag)=\act}$.

\begin{definition}
\label{def:action-abstraction} 
The \defstyle{action profile abstraction} is the function $\abst: \actprofset{} \ra \actassset$ where $\abst(\actprof{})(\actmap(\act)) = |\actprof{}^{-1}(\act)|$ for all $\actprof{} \in \actprofset{}$ {and $\act \in \actset^+$}.
\end{definition}

The function $\abst{}$ partitions the set $\actprofset{}$ into equivalence classes of action profiles having the same abstraction 
{that is, two action profiles $\actprof{1}$ and $\actprof{2}$ belongs to the same equivalence class iff 
$\abst(\actprof{1}) = \abst(\actprof{2})$.}

We now introduce the abstract models of our framework.

\begin{definition}
\label{def:HDMAS} 
  A \defstyle{homogeneous dynamic MAS (\oursys)} is a structure $\dmas = \tuple{\agset, \actset^+, \States, \actav, \transrel, \atprop, \labf}$ where:
  \begin{itemize}
    \item $\agset = \set{\ag_1, \ag_2, \ldots}$ is the countable set of \defstyle{agents}.

  \item $\actset^+$ is the set of \defstyle{action names};
  
  \item $\States$ is a set of \defstyle{states}\footnote{Note that $\States$ is not required in general to be finite, but some of our technical results will assume finiteness.};
  
\item $\actav: \States \ra \wp(\actset^+)$ is the \defstyle{action availability function}, that assigns to every state $\sts$ the set of actions $\actav(\sts)$ available (to all agents) at $\sts$, and is  
such that $\varepsilon \in \actav(\sts)$;
  
\item $\transrel: \States \times \States \ra \guardset$ is the 
\defstyle{transitions guard} function, labelling possible transitions between states with  guards such that:
   \begin{itemize}
   \item $\varsetf{\transrel(\sts, \sts')} \subseteq \actmap[\actav(\sts)]$ for each $\sts, \sts' \in \States$ \ 
 (the guards at each state only involve action counters corresponding to actions available at that state),
 
\item and, for each $\sts \in \States$ and for each $\actass \in \actasssetr{\actmap[\actav(s)]}$, there exists a unique $\sts' \in \States$ such that $\actass \models \transrel(\sts, \sts')$ \  
(every possible action distribution over the set of actions available at the current state determines a unique transition). 
\end{itemize}
  
  \item $\atprop = \set{p_1, p_2, \ldots}$ is a finite set of \defstyle{atomic propositions};
  \item $\labf : \States \ra \wp(\atprop)$  is a \defstyle{labelling function}, assigning to any state $\sts$ the set of atomic propositions that are true at $\sts$.
    \end{itemize}    
  \end{definition}

\input {fig-example-fortress}

\input{fig-example-bis}

  
  \begin{example} \label{ex:fortress}
    {The fortress example presented in the introduction, with $k=3$ entry points, can be modeled as a \oursys as follows. The set $\States$ contain two states only, displayed as circles in Figure~\ref{fig:example-fortress}: $\sts_1$ and $\sts_2$ represents respectively the fortress being under control of the defenders or being captured. Next, we have two actions for each entry point $A_i$: one modelling the defensive action $\act_i$, and the other the attacking action $\overline{act}_i$ for $A_i$; therefore $\actset^+ = \set{\act_1, \overline{act}_1, \act_2, \overline{act}_2, \act_3, \overline{act}_3, \noop}$, with $\actmap(\act_i)=\actvar_i$ and $\actmap(\overline{act}_i)=\overline{\actvar}_i$ for $i \in \set{1, 2, 3}$. All of them are allowed in $\sts_1$ and none of them in $\sts_2$, formally: $\actav(\sts_1) = \actset^+$ and $\actav(\sts_2) = \set{\noop}$.
 The guards $\guard_1, \guard_2$ are listed next to the picture, and an arrow is drawn from $s_i$ to $s_j$ and labeled with $g_k$ iff $\delta(s_i, s_j)=\guard_k$. Formula $\guard_1$ guards transition from $\sts_1$ to $\sts_2$ and therefore it defines when the fortress is captured. This happens when, for each of the entry point $A_i$ with $i \in \set{1, 2, 3}$, one of two conditions hold: 1) the number of defenders $\actvar_i$ is less than $m_i$ or 2) it is less than $M_i$ and also less than the number of attackers $\overline{\actvar}_i$. If this is not the case, the defenders hold the fortress (loop in $\sts_1$) but once is conquered, it remains so regardless of the actions performed ($\guard_2$ is a tautology).
The label of each state, as defined by the labelling function,
 is given next to it. We only have one atomic proposition, $\mathit{captured}$, false in $\sts_1$ and true in $\sts_2$, therefore $\labf(\sts_1) = \emptyset$ and $\labf(\sts_2) = \set{captured}$.
    }
    \end{example}

\begin{example} 
\label{ex:1}
{A more abstract example is given in
Figure \ref{fig:example}, which will be used to illustrate some technical points and the model checking algorithm later.}
      The set of actions is $\actset = \set{\act_1, \act_2, \act_3}$ and the action availability function is defined by $\actav(\sts_1)=\actav(\sts_3)=\actav(\sts_4)=\actset^+$, $\actav(\sts_2)=\set{\act_1, \act_3, \noop}$, $\actav(\sts_5)=\set{\act_2, \act_3, \noop}$ and $\actav(\sts_6)=\set{\act_1, \noop}$.
Lastly, the labelling function is defined as: $\labf(\sts_1)=\emptyset$, $\labf(\sts_2)=\labf(\sts_3)=\labf(\sts_4)=\set{p}$ and $\labf(\sts_5)=\labf(\sts_6)=\set{q}$.
\end{example}

 The restriction on $\transrel$ ensures that for any number of agents and their action profile  of available actions, the next state is uniquely defined. Thus, the dynamics of the  system in terms of possible state transitions is fully determined symbolically by the transitions guard function $\transrel$, as defined formally below.

\begin{definition}
\label{def:transition}
  Given a \oursys $\dmas$, a \defstyle{transition} in $\dmas$ is a triple $(\sts, \actprof{}, \sts')$, where $\sts, \sts'  \in \States$ and $\actprof{} \in \actprofset{}$, such that: \\
  1) each agent $\ag$ performs an available action: $\actprof{}(\ag) \in \actav(\sts)$;
  \\
  2) the abstraction $\abst(\actprof{})$ satisfies the (unique) guard that labels the transition from $\sts$ to $\sts'$, i.e., $\abst(\actprof{}) \models \transrel(\sts, \sts')$.
\end{definition}

Since transitions only depend on the abstractions of the action profiles, that is, on action distributions, it is immediate to see that actions profiles with the same abstraction, applied at the same state, lead to the same successor state. Formally, the following holds. 

\begin{lemma}
\label{fact1} 
  Given a \oursys $\dmas$ as above,  
for every $\sts, \sts' \in \States$,  and every $\actprof{1}, \actprof{2} \in \actprofset{}$, if $\abst(\actprof{1}) = \abst(\actprof{2})$, then $(\sts, \actprof{1}, \sts')$ is a transition in $\dmas$ iff $(\sts, \actprof{2}, \sts')$ is a transition in $\dmas$.
\end{lemma} 

Lemma \ref{fact1} enables us to define the transition function\footnote{
We remark that the assumption of determinism of \oursys is common in the study of multi-agent systems, because non-determinism can be settled easily by the actions of  a fictitious new agent (Nature). Intuitively, one can transform a nondeterministic \oursys to a deterministic one by adding actions that resolve the non-determinism. Then specifications can be translated from the latter to the former by adding controllable or non-controllable agents that could execute these actions.
} of $\dmas$ directly on action distributions, rather than on action profiles. 

\begin{definition}
\label{def:transition-function}
  Let $\dmas$ be a \oursys. The \defstyle{transition function of $\dmas$} is the partial mapping  $\exptransf: \States \times \actassset \dashrightarrow \States$ defined as follows. For each $\sts \in \States$ and $\actass \in \actassset$,
 \defstyle{the outcome state $\exptransf(\sts, \actass)$} of $\actass$ at $\sts$  
is defined and equal to $\sts' \in \States$ 
 iff there exists $\actprof{} \in \actprofset{}$ such that $(\sts, \actprof{}, \sts')$ is a transition and $\abst(\actprof{}) = \actass$; otherwise $\exptransf(\sts, \actass)$ is undefined.
\end{definition}

Infinite sequences of successor
states will be called `plays'. Formally, a \defstyle{play} is a sequence $\pth = \sts_0, \sts_1, \ldots$ in $\States^{\omega}$, such that for every \defstyle{stage} (of the play) 
$i \in \nat$, there is $\actass_i \in \actassset$ such that 
$\exptransf(\sts_i, \actass_i) = \sts_{i+1}$. 
We denote by $\pth[i]$ the state of the $i$-th stage of the play, for each $i \in \nat$.

Since transitions from a given state $\sts$ are defined only for action profiles that assigns to all agents only actions that are available at $\sts$, we call these \defstyle{available action profiles in $\sts$}.  We formally define for each state $\sts \in \States$ the set of available action profiles in $\sts$ as 
\[\actprofset{\sts} = \set{\actprof{} \in \actprofset{} \mid \actprof{}(\ag) \in \actav(\sts) 
\ \mbox{for each} \  \ag \in \agset}.\] 

More generally, for each set of agents $\setagt \subseteq \agset$ we define likewise the set of \defstyle{joint actions for $\setagt$ available in $\sts$} as 
\[\actprofsetr{\sts}{\setagt} = \set{\actprof{\setagt} \in \actprofset{\setagt} \mid \actprof{\setagt}(\ag) \in \actav(\sts) 
\ \mbox{for each} \  \ag \in \setagt}.\]   
where $\actprofset{\setagt}$ denotes (with a mild abuse of notation) the set of all possible joint actions for $\setagt$.

Next, we define a positional strategy for a given coalition of agents $\setagt$  
as a mapping that assigns to each state $\sts$ an available joint action for $\setagt$. 

\begin{definition}
\label{def:joint-strategy}
  Let $\setagt$ be a (possibly empty) set of agents and $\dmas$ be a \oursys with a state space $\States$.  A \defstyle{joint (positional) strategy for the coalition $\setagt$} is a function $\stratgn{\setagt} : \States \ra \actprofsetr{}{\setagt}$ such that   $\stratgn{\setagt}(s) \in \actprofsetr{\sts}{\setagt}$ for each $\sts \in \States$. The empty coalition has only one  joint strategy $\stratgn{\emptyset}$, assigning the empty joint action at every state. 
  \end{definition}

Hereafter we assume that at every stage of the play representing the evolution of the system, the set of all currently present agents is partitioned into two: the set of 
 \emph{controllable agents}, denoted by $\coal$, and the set of 
 \emph{uncontrollable agents}, denoted by $\opp$. Neither of these subsets (and their sizes) is fixed initially, nor during the play, but each of them can vary at each transition round.

\begin{definition}
\label{def:outcome-set}
  Let $\dmas$ be a \oursys, $\sts \in \States$ be a state in it, $\coal, \opp\subseteq \agset$ be the respective current sets of controllable and uncontrollable agents, and   
let $\actprof{\coal} \in \actprofsetr{\sts}{\coal}$.
The  \defstyle{outcome set of $\actprof{\coal}$ at $\sts$} is defined as follows: 
 \[\outf(\sts, \actprof{\coal}, \opp) :=  \{ \sts' \in \States \mid 
   \sts' \!=\! \exptransf(\sts, \abst(\actprof{})) 
\,   \mbox{for\! some} \, \actprof{} \in   \actprofsetr{\sts}{(\coal \cup \opp)}   
 \,   \mbox{such\! that }  \actprof{} |_{\coal} = \actprof{\coal}\}.
   \]
  
Respectively, given a joint strategy $\stratgn{\coal}$ for $\coal$ 
we define the \defstyle{set of outcome plays of $\actprof{\coal}$ at $\sts$ (against $\opp$)} as \\ 
 \[ \outf(\sts, \stratgn{\coal}, \opp) := \big\{ \pth = \sts_0, \sts_1, ... \mid 
      \sts_0=\sts \mbox{ and for all } i \in \nat 
     \mbox{ there exists } \actprof{i} \in \actprofsetr{\sts_i}{(\coal \cup \opp)}  
     \]
     \[ 
     \hspace{46mm}   \mbox{ such that }  \actprofr{i}{\coal} = \stratgn{\coal}(\sts_i) \mbox{ and } 
 \exptransf(\sts_i, \abst(\actprof{i})) = \sts_{i+1} 
\big\}.  
   \]
\end{definition}

\medskip

The abstraction $\abst$, although defined on actions profiles, is readily extended over joint actions and naturally specifies an equivalence relation between them: two joint actions are equivalent whenever their abstraction is the same. Likewise for joint strategies, as the next definition formalizes.

\begin{definition}
\label{def:equivalence}
Let $\dmas$ be a \oursys, 
 $\coal_1, \coal_2 \subseteq \agset$
 and $\actprof{\coal_1}, \actprof{\coal_2}$ be respective joint actions for $\coal_1$ and 
 $\coal_2$.
We say that $\actprof{\coal_1}$ and $\actprof{\coal_2}$  are \defstyle{equivalent}, denoted 
$\actprof{\coal_1}~\equiv~\actprof{\coal_2}$, if 
$\abst(\actprof{\coal_1}) = \abst(\actprof{\coal_2})$.

Likewise,  we say that joint strategies $\stratgn{\coal_1}$ and $\stratgn{\coal_2}$ 
are \defstyle{equivalent}, denoted $\stratgn{\coal_1} \equiv \stratgn{\coal_2}$ if they prescribe equivalent joint actions for $\coal_1$ and $\coal_2$ at every state.
\end{definition}

Note that if $\actprof{\coal_1} \equiv \actprof{\coal_2}$ then  $|\coal_1| =  |\coal_2|$ and 
$\actprof{\coal_1}$ and $\actprof{\coal_2}$
produce the same outcome sets.  

\begin{lemma}
\label{lem:equiv} 
Let $\dmas$ be a \oursys and $\coal_1, \coal_2,\opp_1, \opp_2 \ \subseteq \agset$ be such that, 
 $|\coal_1| =  |\coal_2|$,  $|\opp_1| =  |\opp_2|$,  $\coal_1 \cap \opp_1 = \emptyset$, and $\coal_2 \cap \opp_2= \emptyset$. 
Then:  

\begin{enumerate}
\item For any $\sts \in \States$, if $\actprof{\coal_1}$ and $\actprof{\coal_2}$ are two equivalent joint actions available at $\sts$, respectively for $\coal_1$ and $\coal_2$, then 
 $\outf(\sts, \actprof{\coal_1},\opp_1) = \outf(\sts, \actprof{\coal_2},\opp_2)$.

\item If $\stratgn{\coal_1}$ and $\stratgn{\coal_2}$ are two equivalent  joint strategies in $\dmas$,  respectively for  $\coal_1$ and $\coal_2$, then for each $\sts \in \States$, 
 $\outf(\sts, \stratgn{\coal_1},\opp_1) = \outf(\sts, \stratgn{\coal_2},\opp_2)$.
 \end{enumerate}
\end{lemma}

\begin{proof}
 (1) Let $\sts' \in \outf(\sts, \actprof{\coal_1},\opp_1)$. 
 Then  $\sts' = \exptransf(\sts, \abst(\actprof{1}))$ 
for some  $\actprof{1} \in \actprofsetr{\sts}{(\coal_1 \cup \opp_1)}$   
such that $\actprof{1} |_{\coal_1} = \actprof{\coal_1}$.
 Fix a bijection $h: \coal_2 \ra \coal_1$. 
It can be extended to a bijection $f: \agset \ra \agset$, such that 
 $f[\opp_2] = \opp_1$. 
Define $\actprof{2} \in \actprofsetr{\sts}{(\coal_2 \cup \opp_2)}$ 
  so that $\actprof{2}(\ag) := \actprof{1}(f(\ag))$. 
 Clearly, $\abst(\actprof{2}) = \abst(\actprof{1})$. 
Also $\actprof{2}|_{\coal_2}  = \actprof{1}|_{f[\coal_2]}$ 
as $f[\coal_2] = \coal_1$, hence 
$\abst(\actprof{2}|_{\coal_2})  = \abst(\actprof{1}|_{\coal_1}) = 
\abst(\actprof{\coal_1}) = \abst(\actprof{\coal_2}) $ (since 
  $\actprof{\coal_1} \equiv \actprof{\coal_2}$). 
   Therefore, we obtain that $\sts' = \exptransf(\sts, \abst(\actprof{2}))  
   \in \outf(\sts, \actprof{\coal_2},\opp_2)$. 
  Thus,  $\outf(\sts, \actprof{\coal_1},\opp_1) \subseteq \outf(\sts, \actprof{\coal_2},\opp_2)$. 
  The proof of the converse inclusion is completely symmetric. 
  
 \medskip 
 (2) The claim follows easily by using (1). Indeed, every play $\pth = \sts_0, \sts_1, ... $ in $\outf(\sts, \stratgn{\coal_1},\opp_1)$ can be generated step-by-step as a play in $\outf(\sts, \stratgn{\coal_2},\opp_2)$, by using the equivalence of $\stratgn{\coal_1}$ and $\stratgn{\coal_2}$ and  applying (1) at every step of the construction. We leave out the routine details. 
Thus, $\outf(\sts, \stratgn{\coal_1},\opp_1) \subseteq \outf(\sts, \stratgn{\coal_2},\opp_2)$. 
Again, the converse inclusion is completely symmetric. 
\end{proof}

\medskip

We now prove that, as expected, the outcome sets from joint actions and strategies do not depend on the actual sets of controllable and uncontrollable agents, but only on their sizes.

\begin{lemma}
\label{lem:act-abstraction}
  Let $\dmas$ be a \oursys, $\sts \in \States$,  with 
  $\coal, \opp\subseteq \agset$ be the respective current sets of controllable and uncontrollable agents (hence, assumed disjoint), and let 
  $\actprof{\coal} \in \actprofsetr{\sts}{\coal}$ be an available joint action for $\coal$ at $\sts$. 
Then for every $\coal' \subseteq \agset$ such that $|\coal'| = |\coal|$ there exists an available joint action $\actprof{\coal'}$ for $\coal'$ at $\sts$, 
such that for every $\opp' \subseteq \agset$ where $\coal' \cap \opp' = \emptyset$, if $|\opp'| = |\opp|$, then 
$\outf(\sts, \actprof{\coal'}, \opp') = \outf(\sts, \actprof{\coal}, \opp)$.
\end{lemma}

\begin{proof}
Fix any $\coal' \subseteq \agset$ such that $|\coal'| = |\coal|$. 
 Take a bijection $h: \coal' \ra \coal$. 
It transforms canonically the joint action $\actprof{\coal}$ to a joint action $\actprof{\coal'}$ available at $\sts$, defined by 
 $\actprof{\coal'}(\ag) := \actprof{\coal}(h(\ag))$. Clearly, $\abst(\actprof{\coal'}) = \abst(\actprof{\coal})$. 
 Hence, by Lemma \ref{lem:equiv},  
 $\outf(\sts, \actprof{\coal'}, \opp') = \outf(\sts, \actprof{\coal}, \opp)$ for every 
  $\opp' \subseteq \agset$ such that $\coal' \cap \opp' = \emptyset$ and 
  $|\opp'| = |\opp|$. 
\end{proof}

Lemma \ref{lem:act-abstraction} easily extends to joint strategies, as follows. 

\begin{lemma}
\label{lem:str-abstraction}
  Let $\dmas$ be a \oursys, $\sts \in \States$,  with 
  $\coal, \opp\subseteq \agset$ be the respective current (disjoint) sets of controllable and uncontrollable agents, and let $\stratgn{\coal}$ be a joint strategy for \ 
  $\coal$. 
Then for every $\coal' \subseteq \agset$ with $|\coal'| = |\coal|$ there exists a joint strategy $\stratgn{\coal'}$ 
such that for every $\opp' \subseteq \agset$ where $\coal' \cap \opp' = \emptyset$, 
 if $|\opp'| = |\opp|$, then $\outf(\sts, \stratgn{\coal'}, \opp') = \outf(\sts, \stratgn{\coal}, \opp)$.
\end{lemma}

\begin{proof}
The argument is similar to the previous proof. 

Fix any $\coal' \subseteq \agset$ such that $|\coal'| = |\coal|$. 
 Take a bijection $h: \coal' \ra \coal$. 
It transforms canonically the joint strategy $\stratgn{\coal}$  to a joint strategy  
$\stratgn{\coal'}$, defined by 
 $\stratgn{\coal'}(\sts)(\ag) := \stratgn{\coal}(\sts)(h(\ag))$. 
 Clearly, $\abst(\stratgn{\coal'}(\sts)) = \abst(\stratgn{\coal}(\sts))$ for every state $\sts$, hence $\stratgn{\coal} \equiv \stratgn{\coal'}$.   
Therefore, by Lemma \ref{lem:equiv},  
$\outf(\sts, \stratgn{\coal'}, \opp') = \outf(\sts, \stratgn{\coal}, \opp)$ for every 
  $\opp' \subseteq \agset$ such that $\coal' \cap \opp' = \emptyset$ and 
  $|\opp'| = |\opp|$. 
  \end{proof}

Lemmas \ref{lem:act-abstraction} and \ref{lem:str-abstraction} essentially say that the strategic abilities in a \oursys are determined not by the concrete sets of controllable and uncontrollable agents, but only by their respective sizes. This justifies abstracting the notions of coalitional actions and strategies in terms of action profile abstractions, to be used thereafter in our semantics and verification procedures.

\begin{definition}
\label{def:abstr-action}
\label{def:abstr-strategy}
\label{def:abstr-outcome}
Let $\dmas$ be a \oursys and $C,N \in \nat$. 

\medskip
1.1. An \defstyle{abstract joint action for a coalition of $C$ agents at state $\sts \in \States$} is an action distribution 
 $\actass_{C} \in \actasssetn{C}$ such that $\dom(\actass_{C}) = \actmap[\actav(\sts)]$ \ 
(recall notation from Definition \ref{def:act-dist}).   

 Thus, an abstract joint action for a given coalition at state $\sts$ prescribes for each action available at $\sts$ how many agents from the coalition take that action.

\smallskip
1.2. The \defstyle{outcome set of states of the abstract joint action $\actass_{C}$ 
  of $C$ controllable agents against $N$ uncontrollable agents at $\sts$} is the set of states
\[ \outf(\sts, \actass_{C}, N) := 
    \big\{ \sts' \in \States \mid \sts' = \exptransf(\sts, \actass_{C} \oplus  \actass_{N})  
    \mbox{ for some } \actass_{N} \in \actasssetn{N}  
         \]
     \[ 
 ~    \hspace{66mm}
      \mbox{ such that } \dom(\actass_{N}) = \actmap[\actav(\sts)]  
 \big\}.  
  \]
 
\smallskip
2.1. An \defstyle{abstract (positional) joint strategy for a coalition of $C$ agents} is a function $\astratg_{C}: \States \ra \actasssetn{C}$ such that for each $\sts \in \States$,
$\astratg_{C}(\sts)$ is an abstract joint action such that 
$\dom(\astratg_{C}(\sts)) = \actmap[\actav(\sts)]$.

\smallskip
2.2. The \defstyle{outcome set of plays of an abstract  joint strategy $\astratg_{C}$   
  of $C$ controllable agents 
  against $N$ uncontrollable agents} is the set of plays
  \[
  \outf(\sts, \astratg_{C}, N) := 
\big\{ \pth = \sts_0, \sts_1, ... \mid  \sts_0=\sts  \mbox{ and for all } i \in \nat  
 \mbox{ there is } \actass_{i} \in \actasssetn{N} 
 \]
   \[
   ~    \hspace{26mm} 
 \mbox{ such that } \dom(\actass_{i}) = \actmap[\actav(\sts)]  \mbox{ and }  
\Delta(\sts_i, \astratg_{C}(\sts_i) \oplus \actass_{i} ) = \sts_{i+1}
\big\}.
 \]
\end{definition}


%% file: fig-example-fortress.tex
\tikzset{every state/.style={minimum size=4mm}}
\tikzset{->,>=stealth',shorten >=1pt}

\begin{figure}
  \begin{minipage}{0.5\textwidth}
      \begin{center}
    \begin{small}
  
\begin{tikzpicture}[node distance=1.5cm]
  
  \node[state] (1) {$\sts_1$};
  \node () [below =0.1cm of 1] {$\set{}$};

  \node[state] (2) [right of=1] {$\sts_2$};
  \node () [below =0.05cm of 2] {$\set{\mathit{captured}}$};

  \path

  (1) edge [loop left] node[left]{$\neg \guard_1$}  (1)
  (1) edge node[above]{$\guard_1$} (2)

  (2) edge [loop right] node[right]{$\guard_2$} (2)
  ;
\end{tikzpicture}
\end{small}
\end{center}
\end{minipage} \begin{minipage}{0.5\textwidth}
  \begin{center}
    \begin{small}
      \begin{align*}
        \guard_1 := {} & \bigwedge_{i \set{1, 2, 3}} \actvar_i < m_i \lor (\actvar_i < M_i \land \actvar_i < \overline{\actvar}_i) \\
        \guard_2 := {} & \actvar_1 = \actvar_1
\end{align*}
    \end{small}
    \end{center}
\end{minipage} 
\caption{The fortress example modelled as a \oursys.} 
\label{fig:example-fortress}
\end{figure}
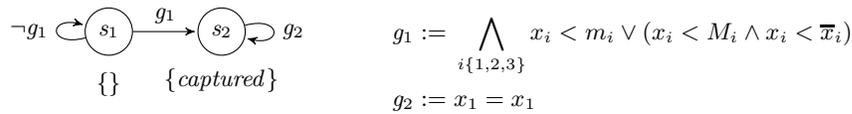

%% file: fig-example-bis.tex
\tikzset{every state/.style={minimum size=4mm}}
\tikzset{->,>=stealth',shorten >=1pt}

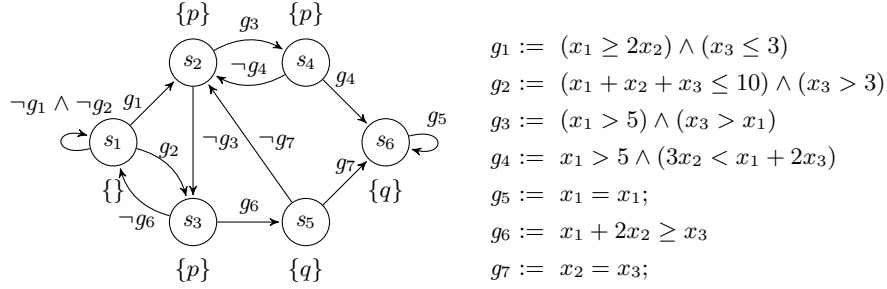
\begin{figure}
  \begin{minipage}{0.5\textwidth}
  \begin{center}
    \begin{small}
\begin{tikzpicture}[node distance=1.5cm]
  
  \node[state] (1)                    {$\sts_1$};
  \node () [below =0.05cm of 1] {$\set{}$};

  \node[state]         (2) [above right of=1] {$\sts_2$};
  \node () [above=0.05cm of 2] {$\set{p}$};

  \node[state]         (3) [below right of=1] {$\sts_3$};
  \node () [below=0.05cm of 3] {$\set{p}$};

  \node[state]         (4) [right of=2] {$\sts_4$};
  \node () [above=0.05cm of 4] {$\set{p}$};

  \node[state]         (5) [right of=3] {$\sts_5$};
  \node () [below=0.05cm of 5] {$\set{q}$};

  \node[state]         (6) [below right of=4] {$\sts_6$};
  \node () [below =0.05cm of 6] {$\set{q}$};

  \path

  (1) edge [loop left] node[above =0.25cm of 1]{$\neg \guard_1 \land \neg \guard_2$}  (1)
  (1) edge node[left]{$\guard_1$} (2)
  (1) edge [bend left] node[above]{$\guard_2$}  (3)

  (2) edge node[right]{$\neg \guard_3$} (3)
  (2) edge [bend left] node[above]{$\guard_3$} (4)

  (3) edge [bend left] node[below =0.1cm of 3]{$\neg \guard_6$} (1)
  (3) edge node[above]{$\guard_6$} (5)

  (4) edge [bend left] node[above]{$\neg \guard_4$} (2)
  (4) edge node[above=0.1cm of 4]{$\guard_4$} (6)

  (5) edge node[right]{$\neg \guard_7$} (2)
  (5) edge node[above]{$\guard_7$} (6)

  (6) edge [loop right] node[above = 0.1cm of 6]{$\guard_5$} (6)
  ;
\end{tikzpicture}
\end{small}
\end{center}
\end{minipage}\begin{minipage}{0.5\textwidth}
  \begin{center}
    \begin{small}
\begin{align*}
  \guard_1 := {} & \ (\actvar_1 \geq 2 \actvar_2) \land (\actvar_3 \leq 3)\\
  \guard_2 := {} & \ (\actvar_1 + \actvar_2 + \actvar_3 \leq 10) \land (\actvar_3 > 3) \\
  \guard_3 := {} & \ (\actvar_1 > 5) \land (\actvar_3 > \actvar_1) \\
  \guard_4 := {} & \ \actvar_1 > 5 \land (3 \actvar_2 < \actvar_1 + 2 \actvar_3) \\
  \guard_5 := {} & \ \actvar_1 = \actvar_1; \\
  \guard_6 := {} & \ \actvar_1 + 2 \actvar_2 \geq \actvar_3 \\
  \guard_7 := {} & \ \actvar_2 = \actvar_3;                 
\end{align*}
\end{small}
\end{center}
\end{minipage}

  \caption{An abstract example of a \oursys.}
\label{fig:example}
\end{figure}

%% file: logic.tex
\section{Logic for specification and verification of HDMAS} 
\label{sec:logic}

We now introduce a logic $\lang$ for specifying and verifying properties of \oursys, based on the Alternating-time Temporal Logic ATL. It features a strategic operator that expresses the ability of a set of controllable agents to guarantee the satisfaction a temporal objective, regardless of the actions taken by the set of uncontrollable agents. As shown in the previous section, such ability only depends on the sizes of these sets. Therefore, our strategic operator $\coop{\ast, \ast}$ takes two arguments: the first one represent the number of controllable agents and the second  -- the number of uncontrollable agents currently present in the system. Intuitively, a formula of the kind $\coop{C, N} \chi$, with $C, N \in \nat$ and $\chi$ being a (path) formula of $\lang$ specifies the property: \\ 
``\textit{A coalition of $C$ controllable agents has a joint  strategy to guarantee satisfaction of the objective $\chi$ against $N$ uncontrollable agents on every play consistent with that strategy}''.  

Each of the arguments $C$ and $N$ may be a concrete number, a parameter, or a variable that can be quantified over. Parameters are free variables that cannot be quantified over, which gives extra expressiveness of the language, because some syntactic restrictions will be imposed on the variables.

\subsection{Formal syntax and semantics} 
\label{subsec:synsem} 

We now fix a set of atomic propositions $\Phi = \{p_1,p_2,....\}$,
a set of two special variables $\agvarset = \set{\agvar_1, \agvar_2}$, ranging over $\nat$, which we call \defstyle{agent counters}. These will represent the numbers of controllable and uncontrollable agents respectively, and can be quantified over. 
We also fix a set of  \defstyle{agent counting parameters}\footnote{The role of parameters is mostly auxiliary, just like in algebra (or in first-order logic) and will be discussed further.}
 $\parset = \set{\parvar_1, \parvar_2, \ldots}$, again ranging over $\nat$, 
and define the set of \defstyle{terms}\footnote{To avoid cluttering the notation, we will identify here natural numbers with their numerals.} as $\termset = \agvarset \cup \parset \cup \nat$.
These will be used as arguments of the strategic operators in the logical language defined below.  

\begin{definition}\label{def:syntax} 
The logic $\lang$ has two sorts of formulae, defined by mutual induction with the following grammars, where {free} (and {bound}) occurrences of variables
are defined like in first-order logic (FOL): 

\smallskip
\defstyle{Path formulae}: \ \  
$ \chi ::= {}  \atlx \varphi \mid 
\atlg \varphi \mid \psi \atlu \varphi $, \\ 
where $\varphi, \psi$ are state formulae. 

\smallskip
\defstyle{State formulae}:   \\
$\varphi ::= {}  \top \mid p \mid \neg \varphi \mid 
 (\varphi \land \varphi)  \mid (\varphi \lor \varphi)  \mid
     \coop{\term_1, \term_2}{ \chi} \mid \forall \agvar \varphi  \mid \exists \agvar \varphi$
          \smallskip 
\\
 where $p \in \Phi$,
   $\term_1 \in \termset \setminus \{\agvar_2\}$, 
   $\term_2 \in \termset \setminus \{\agvar_1\}$,    
   $\agvar \in \agvarset$,  
   and $\chi$ is a path formula. 
   The cases of $\forall \agvar \varphi$ and  $\exists \agvar \varphi$ are subject to the following syntactic constraint:  all free occurrences of $\agvar$ in $\varphi$ must have a positive polarity, viz. must be in the scope of an even number of negations.

The propositional connectives $\bot, \to, \ifff$
 are defined as usual. 
 Also, we define
 ${\atlf \psi} :=   
 {\top \atlu \psi}$.  
\end{definition}

\begin{remark}  Some remarks on the 
formulae in $\lang$ are in order: 

\begin{enumerate}
\item 
Note that $\agvar_1$ can only occur in the first position  of $\coop{\term_1, \term_2}$ and $\agvar_2$ can only occur in the second position.  
However, the same parameter $z$ may occur in both positions and this is one reason to allow the use of parameters, as the model checking algorithm will treat them uniformly. 

\item The restriction for quantification only over positive free occurrences of variables is imposed for technical reasons. By using the duality of $\forall$ and $\exists$, that restriction can readily be relaxed to the requirement all free occurrences of the quantified variable to be of \emph{the same} polarity (all positive, or all negative). Further relaxation, allowing both positive and negative occurrences under some restrictions, is possible, but it would complicate further the syntax and the model checking algorithm,  without making an essential contribution to the \emph{useful} expressiveness of the language. Indeed, one can argue that, if a formula is to make a meaningful claim about the strategic abilities of the coalition of controllable agents which is quantified over the number of these agents, then it is natural to assume that the controllable coalition appear only in positive context in that claim\footnote{For instance, $\exists y_1 \coop{y_1, 10} \atlx \neg \exists y_1 \coop{y_1, 5} p$, where $y_1$ appear first positively and then negatively, expresses that ``there exists a coalition of controllable agents that can ensure against 10 uncontrollable agents that at the next step there is no coalition of controllable agents that can ensure against 5 uncontrollable agents the truth of $p$'', which is a rather unusual specification in any practical context.}.

\item Some additional useful syntactic restrictions can be imposed, which (as it will be shown in the next section) do not essentially restrict the expressiveness of the language. They lead to the notion of `normal form', to be introduced shortly.   
\end{enumerate}
\end{remark}

 Hereafter, by \lang-formulae we will mean, unless otherwise specified, state formulae of \lang, whereas we will call the path formulae in $\lang$ \defstyle{temporal objectives}. In particular,  
for any \lang-formula $\phi$ of the type $ \coop{\term_1, \term_2}  \chi$, the path subformula $\chi$ is called the \defstyle{temporal objective of $\phi$}.

Some examples of $\lang$ formulae: 
\begin{itemize}
\item 
{with reference to the fortress example: 

\begin{itemize}
\item $\coop{C, N_1} \atlx \coop{C, N_2} \atlx \coop{C, N_3} \atlx \neg \mathit{captured}$, with $N_1 < N_2 < N_3$ and $N_i \in \nat$ for $i = \set{1, 2, 3}$, says that there is a strategy for $C \in \nat$ defenders to hold the fortress for three days against an increasing number of attackers.

\item
$\exists \agvar_1 \coop{\agvar_1, N} \atlg \neg \mathit{captured}$ expresses that there is a number $\agvar_1$ of defenders that have a strategy to hold the fortress forever against $N$ many invaders.  

\item
$\forall \agvar_2 \coop{C, \agvar_2} \atlg \neg \mathit{captured}$ expresses that, for any number of invaders $\agvar_2$, there is a strategy for $C$ defenders to hold the fortress  forever against $\agvar_2$ invaders\footnote{We note that such a strategy would generally depend on $\agvar_2$. However, given the monotonicity properties that our logic enjoys (proved in the next section), it turns out that the above reading of the formula is equivalent on finite models to ``there exists a strategy for $C$ defenders to hold the fortress agains any number $\agvar_2$ of attackers''.}.

\item $ \forall y_2 \exists y_1 \coop{ y_1, y_2} \atlg \lnot \mathit{captured}$  expresses that for any number ($y_2$) of invaders there is a number ($y_1$)  of defenders  who have a joint strategy to  hold the fortress forever.   
\end{itemize}
}

\item lastly, an abstract example with nesting of strategic operators and quantifiers: \\
$\coop{z_2, z_2} \atlx p \lor \exists y_1 (\coop{y_1, z_1} \atlf \coop{y_1, y_2} \atlx \lnot p  \land 
\lnot \forall y_2
\coop{z_1, y_2} \atlx  
\lnot \coop{y_1, z_2} p \atlu q)$, \\   
for $z_1,z_2 \in \parset$.
\end{itemize}

\medskip
The semantics of $\lang$ is based on the standard, positional strategy semantics of ATL (cf \cite{AHK-02} or \cite{BGJ15}), applied in \oursys models, but uses abstract joint actions and strategy profiles, rather than concrete ones. 
In order to evaluate formulae that contain free variables and parameters, we use a version of FOL \defstyle{assignment}, here defined as a function $\agass: \termset \ra \nat$, where $\agass(i)=i$ for $i \in \nat$.

\begin{definition}\label{def:semantics}
Let $\dmas$ be a \oursys, $\sts$ be a state and $\agass$ an assignment in it. The \defstyle{satisfaction relation} $\models$ is inductively defined on the structure of $\lang$-formulae as follows:

      \begin{enumerate}
     \item $\dmas, \sts, \agass \models \top$; 
      
\item $\dmas, \sts, \agass \models p$ iff $p \in \labf(\sts)$;

\item $\land$ and $\neg$ have the standard semantics;

\item\label{def:semantics:coop} $\dmas, \sts, \agass \models \coop{\term_1, \term_2}{\chi}$ iff there exists an abstract strategy $\astratg_{C}$ for a coalition of $C = \agass(\term_1)$ agents such that for every play $\pth$ in the outcome set $\outf(\sts, \astratg_{C}, N)$ against $N=\agass(\term_2)$ uncontrollable agents the following hold:
  \begin{enumerate}
\item if $\chi = \atlx \varphi$ then $\dmas, \pth[1], \agass \models \varphi$;
\item if $\chi = \atlg \varphi$ then $\dmas, \pth[i], \agass \models \varphi$ for every $i \in \nat$;
 \item  if $\chi = \varphi_1 \atlu \varphi_2$ then $\dmas, \pth[i], \agass \models \varphi_2$ for some $i \geq 0$ and $\dmas, \pth[j], \agass \models \varphi_1$ for all $0 \leq j < i$;
\end{enumerate}
  
\item $\dmas, \sts, \agass \models \forall \agvar \varphi$ iff $\dmas, \sts, \agassow{\agvar}{m} \models \varphi$ for every $m \in \nat$, 
where the assignment $\agassow{\agvar}{m}$ assigns $m$ to $\agvar$ and agrees with $\agass$ on every other argument. 

\item Likewise for $\dmas, \sts, \agass \models \exists \agvar \varphi$.  
\end{enumerate}

\end{definition}

The notions of \defstyle{validity}
and \defstyle{(logical) equivalence} in $\lang$ are defined as expected, and we will use the standard notation for them, viz. $\models \varphi$ for validity and $\varphi_1 \equiv \varphi_2$ for equivalence. 
We also say that two $\lang$-formulae, $\varphi_1$ and $\varphi_2$ are \defstyle{equivalent in the finite}, denoted $\varphi_1 \fineq \varphi_2$, if $\dmas, \sts,  \agass \models \varphi_1$ iff $\dmas, \sts,  \agass \models \varphi_2$ for any finite \oursys model  
$\dmas$ and state $\sts$ and assignment $\agass$ in $\dmas$.

\begin{remark} Note the following:  
\begin{enumerate}
\item Defining the semantics in terms of abstract joint actions and strategies in the truth definitions of the strategic operators, rather than concrete ones, is justified by Lemmas \ref{lem:act-abstraction} and \ref{lem:str-abstraction}
which imply that the `concrete' and the `abstract' semantics are equivalent.

\item Just like in FOL, the truth of any $\lang$-formula $\varphi$ only depends on the assignment of values to the parameters that occur in $\varphi$ and to the variables that occur free in $\varphi$. In particular, it does not depend at all on the assignment for \emph{closed formulae} (containing no parameters and free variables). In such cases we simply write 
 $\dmas, \sts \models \varphi$.    
 
 \item
 Again, just like in FOL, if $\agvar$ has no free occurrences in $\varphi$, then 
 $\forall \agvar \varphi \equiv \exists \agvar \varphi \equiv \varphi$. 
 Thus, in order to avoid such vacuous quantification, whenever it occurs we can assume that the formula is simplified automatically according to these equivalences. 
\end{enumerate}
\end{remark}

\begin{example}
\label{ex:4.3}  
Consider the \oursys $\dmas$ in Example~\ref{ex:1}. 

\begin{enumerate}
\item The closed formula $\varphi = \coop{7, 5} \atlx p$
is satisfied in state $\sts_1$ of $\dmas$.
Indeed, any abstract joint strategy $\astratg_{7}$ that prescribes $\noop$ to 3 of the  controllable agents ($\astratg_{7} (\sts_1)(\noop)=3$) and $\act_{3}$ to 4 of them ($\astratg_{7} (\sts_1)(\act_{3})=4$) guarantees that guard $\guard_2$ is satisfied, enforcing transition from $\sts_1$ to $\sts_3$. 
\item
 $\dmas, \sts_1 \models \lnot \exists \agvar_1 \coop{\agvar_1, 11} \atlx p$.
  Indeed, for any value of $\agvar_1$ the abstract joint action profile for the uncontrollable agents that prescribes to all of them to perform $\act_3$ falsifies both $\guard_1$ and $\guard_2$, thus forces a loop to $\sts_1$ where $p$ is false.

 \item $\dmas, \sts_4 \models  \coop{7, 4} \atlx (\forall \agvar_2 \exists \agvar_1 \coop{\agvar_1, \agvar_2} \atlg p)$, as we show in Section \ref{sec:MC}.
\end{enumerate}  
\end{example}

\subsection{Normal form and monotonicity properties} 
\label{subsec:normalforms1} 

This is a technically important section, where we define the fragment \langb of normal form formulae of \lang.
The normal form impose essential syntactic restrictions and therefore reduce the expressiveness of the language. However, the key technical result obtained here is that every formula in \lang\  is equivalent \emph{on finite models} to one in \langb. The importance of that result will be discussed further.

  \begin{definition}\label{def:normalform}
A $\lang$-formula $\psi$ is \defstyle{in a normal form} if: 
\begin{enumerate}
\item[(NF1)] There are no occurrences of  $\forall y_1$ or $\exists y_2$ in $\psi$.  

\item[(NF2)] Every subformula $\coop{\term_1, \term_2}  \chi$ of $\psi$ where either $\term_1 = y_1$ or $\term_2 = y_2$ (but not both), such that that variable occurrence is bound in $\psi$, is immediately preceded respectively by $\exists y_1$ or $\forall y_2$. 

\item[(NF3)] Every subformula $\coop{y_1, y_2}  \chi$, where both variable occurrences are
bound in $\psi$, is immediately preceded either by $\forall y_2 \exists y_1$ or $\exists y_1 \forall y_2$.
\end{enumerate}
 \end{definition}
Of the example formulae given after Definition~\ref{def:syntax}, the first two are in normal form, while the last one is not.   

We denote by \langb\ the fragment of \lang\ consisting of all formulae in normal form. 
We can give a more explicit definition of the formulae of \langb, 
by modifying the recursive definition of state formulae of \lang, where
the clauses  $\forall y \varphi$ and  $\exists y \varphi$ are replaced with the following, where 
$\chi$ is a temporal objective:
\begin{align} \label{eq:nf-temporal}
  \begin{split}
	& \exists y_1 \coop{y_1, \term_2}\chi  \mid  
         \forall y_2  \exists y_1 \coop{y_1, y_2}\chi   \mid  \\
        &  \forall y_2 \coop{\term_1, y_2}\chi   \mid  	
        \exists y_1  \forall y_2 \coop{y_1, y_2}\chi  
        \end{split}
\end{align}

The same syntactic constraints as before apply. In addition, in each case above no variable quantified in the prefix of the formula may occur free in $\chi$.

\medskip
   The rest of the section is devoted to prove that every formula in \lang\ is logically equivalent in the finite to one in \langb. That is of crucial importance, as our model checking algorithm works only on \langb\ formulae. Indeed, the fact that quantification in formulae in normal form does not span across multiple temporal objectives enables us to obtain fixpoint characterizations for formulae of the types listed in~(\ref{eq:nf-temporal}) above, presented at the end of this section, in Theorem \ref{lem:fixpoints1}. That, in turn, allows us to retain the basic structure of the recursive model checking algorithm for ATL (cf \cite{AHK-02} or \cite{BGJ15}).  

A first important observation is that the semantics of the strategic operators in  \lang\xspace is \emph{monotonic} with respect to the number of controllable and uncontrollable agents, in a sense formalized in the following lemma.

Hereafter, for a given formula $\varphi$, term $t$ and $k\in \nat$, we denote by $\varphi[k/t]$ the result of  uniform substitution of all \emph{free}\footnote{The constraint to free occurrences is, of course, only relevant when $t$ is a variable.} 
 occurrences of $t$ in $\varphi$ by $k$.
 
\begin{lemma} 
\label{lem:monotonicity} 
For every $\lang$-formula $\varphi$ and a term $t$ 
the following monotonicity properties hold.  

\begin{description}
\item[\cmon] 
Suppose $C,C' \in \nat$ are such that $C'>C$. Then: 

\cmon$^+$:  If all free occurrences of $t$ are positive
and only in first position in strategic operators in $\varphi$ 
then $\models \varphi[C/t] \to \varphi[C'/t]$. 

\smallskip
\cmon$^-$:  If all free occurrences of $t$ are negative and only in first position in strategic operators in $\varphi$ then $\models \varphi[C'/t] \to \varphi[C/t]$. 

\medskip
\item[\nmon]
Suppose $N,N' \in \nat$ are such that $N' <N$. Then: 

\nmon$^+$:  
 If all free occurrences of $t$ are positive 
and only in second position in strategic operators in $\varphi$ then $\models \varphi[N/t] \to \varphi[N'/t]$.

\smallskip

\nmon$^-$:  
If all free occurrences of $t$ are negative and only in second position in strategic operators in $\varphi$ then $\models \varphi[N'/t] \to \varphi[N/t]$.
\end{description}
\end{lemma} 

\begin{proof} ~
\cmon:
Both claims are analogous and we prove both by simultaneous induction on the structure of $\varphi$. We will present the proof for \cmon$^+$, and the claim of \cmon$^-$ will only be needed in the case when $\varphi = \lnot \psi$, proved by using the inductive hypothesis for \cmon$^-$ for $\psi$ and contraposition.  

\smallskip
The inductive cases where the main connective of $\varphi$ is $\land, \lor, \forall, \exists$ are easily proved by using the inductive hypothesis and the monotonicity of each of these logical connectives.  

The only more essential inductive case is $\varphi = \coop{t, t_2} \chi$, where the inductive hypothesis is that the claim of \cmon$^+$ holds for the main state subformulae of $\chi$. 
Note that the semantics of the strategic and temporal operators is argument-monotone, in sense that if $\models \psi \to \psi'$ then 
 $\models  \coop{t, t_2} \atlx \psi \to  \coop{t, t_2} \atlx \psi'$ and 
 $\models  \coop{t, t_2} \atlg \psi \to  \coop{t, t_2} \atlg \psi'$, and likewise for Until. By using that and the inductive hypothesis, we obtain that 
 $\models  \coop{t, t_2} \chi[C/t] \to \coop{t, t_2} \chi[C'/t] $. 
 Therefore,  $\models  \coop{C, t_2} \chi[C/t] \to \coop{C', t_2} \chi[C'/t] $. 
 Thus, it remains to show that 
 $\models  \coop{C, t_2} \chi[C/t] \to \coop{C', t_2} \chi[C/t] $. 
Let $\dmas, \sts,  \agass \models \coop{C, t_2} \chi[C/t]$.   
Let $\astratg_{C}$ be an abstract strategy for $C$ controllable agents such that every play $\pth$ in the outcome set $\outf(\sts, \astratg_{C}, \agass(t_2))$ against $\agass(t_2)$ uncontrollable agents satisfies the temporal objective $\chi[C/t]$. 
Then, since $C' > C$, the strategy $\astratg_{C}$ can be extended to strategy 
$\astratg_{C'}$ whereby the additional $C'-C$ many agents always perform the idle action 
$\noop$. Clearly, $\astratg_{C'}$ ensures that 
$\dmas, \sts,  \agass \models \coop{C', t_2} \chi[C/t]$. 

\medskip
\nmon: The proof is analogous to the one for \cmon, so we only treat the inductive case of $\varphi = \coop{t_1, t} \chi$ for the claim \nmon$^+$. 
Similarly to  the case of \cmon$^+$, it boils down to proving the validity $\models  \coop{t_1, N} \chi[N/t] \to \coop{t_1,N'} \chi[N/t]$. 
Let $\dmas, \sts,  \agass \models \coop{t_1, N} \chi[N/t]$ and let 
$\astratg_{C}$ be an abstract strategy for $C = \agass(t_1)$ controllable agents such that every play $\pth$ in the outcome set $\outf(\sts, \astratg_{C}, N)$ against $N$ uncontrollable agents satisfies the temporal objective $\chi[N/t]$. 
Then the same strategy would ensure $\dmas, \sts,  \agass \models \coop{t_1,N'} \chi[N/t]$ for every $N' < N$, since every joint action of  $N'$ can be lifted to a joint action of  $N$ leading to the same outcome, where the remaining $N-N'$ agents  always perform the idle action $\noop$. 
\end{proof}

A key consequence of the monotonicity properties is that it allows to eliminate some quantifier patterns, in the cases listed in the following lemma.

\begin{lemma} 
\label{lem:equiv1} 
For every term $t$
and temporal objective $\chi$  in \lang,  the following hold.
\begin{enumerate}
\item 
$\forall y_1 \coop{y_1, t} \chi \equiv \coop{0, t} \chi[0/y_1]$; 

\item 
$\exists y_2 \coop{t, y_2} \chi \equiv \coop{t, 0} \chi[0/y_2]$; 

\item 
$\forall y_1 \exists y_2 \coop{y_1, y_2} \chi  \equiv 
\coop{0, 0} \chi[0/y_1,0/y_2]$; 

\item 
$\exists y_2 \forall y_1 \coop{y_1, y_2} \chi  \equiv  
\coop{0, 0} \chi[0/y_1,0/y_2]$;   

\item 
$\forall y_2 \forall y_1 \coop{y_1, y_2} \chi \equiv
\forall y_2 \coop{0, y_2} \chi[0/y_1]$;     

\item 
$\forall y_1 \forall y_2 \coop{y_1, y_2} \chi \equiv 
\forall y_2 \coop{0, y_2} \chi[0/y_1]$; 

\item 
$\exists y_1 \exists y_2 \coop{y_1, y_2} \chi \equiv 
\exists y_1 \coop{y_1, 0} \chi[0/y_2]$; 

\item 
$\exists y_2 \exists y_1 \coop{y_1, y_2} \chi \equiv 
\exists y_1 \coop{y_1, 0} \chi[0/y_2]$.     
\end{enumerate}
\end{lemma}

\begin{proof} ~ 

The logically non-trivial implications of claims 1-6 follow immediately from the polarity constraint in the definition of formulae and Lemma  
\ref{lem:monotonicity}. Claims 7 and 8 follow respectively from claims 5 and 6, by commuting the quantifiers. 

\end{proof}

Lemma \ref{lem:equiv1} shows that the only non-trivial cases of quantifications over formulae of the kind $\coop{\term_1, \term_2} \chi$ are those allowed in normal forms, listed in ~(\ref{eq:nf-temporal}) (after Definition \ref{def:normalform}). 
We will make use of that to re-define the syntax of $\lang$ to suit better our further technical work.  
First we define an \defstyle{admissible quantifier prefix} $\mathcal{Q}$ to be a string of the form $\Q \agvar_i$ or $\Q\agvar_i\Q'\agvar_j$ where $\Q,\Q' \in \set{\exists, \forall}$ and $i,j \in \set{1, 2}$, $i \neq j$.
Now, we re-define the set of state formulas of $\lang$ to be generated by the following modified grammar:
\[\varphi ::= {}  \top \mid p \mid \neg \varphi \mid 
 (\varphi \land \varphi)  \mid (\varphi \lor \varphi)  \mid
     \coop{\term_1, \term_2}{ \chi} \mid \mathcal{Q} \varphi\]
The same positive polarity requirements as before for applying the quantifier prefixes are imposed. Clearly, this grammar is equivalent to the original grammar, i.e. it generates the same set of formulae. In the rest of the paper we adopt the new grammar above.

Next, we define recursively a \defstyle{partial quantifier elimination function} $\pqe$ on path and state formulae $\xi \in \lang$ which produces formulae $\pqe(\xi)$  where all occurrences of subformulae in the left-hand sides of the equivalences in Lemma~\ref{lem:equiv1} are successively replaced with the corresponding right-hand sides.

\medskip
\begin{algorithmic}[1]
   \State \textbf{Input}: a quantifier prefix $\mathcal{Q}$ and a (state or path) formula $\xi$ of \lang

   \State \textbf{Output}: a (state or path) formula $\xi'$ of \lang

   \Procedure{$\pqe$}{$\xi$}

   \Cases{$\xi$}

   \Case{$\top \mid \avarprop$} \Return $\xi$ 

   \Case{$\neg \psi$} \Return $\neg \pqe(\psi)$

   \Case{$\psi_1 \land \psi_2$} \Return $\pqe(\psi_1) \land \pqe(\psi_2)$

   \Case{$\psi_1 \lor \psi_2$} \Return $\pqe(\psi_1) \lor \pqe(\psi_2)$

   \Case{$\coop{\term_1, \term_2} \chi$} \Return $\coop{\term_1, \term_2} \pqe(\chi)$

   \Case{$\mathcal{Q} \varphi$}

   \Cases{$\mathcal{Q} \varphi$}

   \Case{$\forall \agvar_1 \coop{\agvar_1, \term} \chi$}
   \Return $\coop{0, \term} \pqe(\chi)[0/\agvar_1]$

   \Case{$\exists y_2 \coop{t, y_2} \chi$}
   \Return $\coop{t, 0} \pqe(\chi)[0/y_2]$

   \Case{$\forall y_1 \exists y_2 \coop{y_1, y_2} \chi$}
   \Return $\coop{0, 0} \pqe(\chi)[0/y_1,0/y_2]$

   \Case{$\exists y_2 \forall y_1 \coop{y_1, y_2} \chi$}
   \Return{$\coop{0, 0} \pqe(\chi)[0/y_1,0/y_2]$}

   \Case{$\forall y_1 \forall y_2 \coop{y_1, y_2} \chi$}
   \Return $\forall y_2 \coop{0, y_2} \pqe(\chi)[0/y_1]$

   \Case{$\forall y_2 \forall y_1 \coop{y_1, y_2} \chi$}
   \Return $\forall y_2 \coop{0, y_2} \pqe(\chi)[0/y_1]$

   \Case{$\exists y_1 \exists y_2 \coop{y_1, y_2} \chi$}
   \Return $\exists y_1 \coop{y_1, 0} \pqe(\chi)[0/y_2]$

   \Case{$\exists y_2 \exists y_1 \coop{y_1, y_2} \chi$}
   \Return $\exists y_1 \coop{y_1, 0} \pqe(\chi)[0/y_2]$

      \Case{\textbf{default}} \Return $\mathcal{Q} \pqe(\varphi)$ \Comment{Covers all other cases} 
   
   \EndCases

      \Case{$\atlx \psi$} \Return $\atlx \pqe(\psi)$

   \Case{$\atlg \psi$} \Return $\atlg \pqe(\psi)$

   \Case{$\psi_1 \atlu \psi_2$} \Return $\pqe(\psi_1) \atlu \pqe(\psi_2)$

  \EndCases
   \EndProcedure
 \end{algorithmic}

\begin{lemma}\label{def:pqe-equivalence}
Let $\varphi$ be any formula in $\lang$, then $\pqe(\varphi) \equiv \varphi$.
\end{lemma}
\begin{proof}
By induction on the structure of $\varphi$ (using the modified grammar), following the recursive definition of $\pqe$.  The only non-trivial cases are those in lines 12-19 and they use the equivalences in Lemma~\ref{lem:equiv1}.  For instance, let $\varphi = \forall \agvar_1 \coop{\agvar_1, \term} \chi$. By definition, $\pqe(\forall \agvar_1 \coop{\agvar_1, \term} \chi) = \coop{0, \term} \pqe(\chi)[0/\agvar_1]$ and by inductive hypothesis $\pqe(\chi) \equiv \chi$, thus we get $\coop{0, \term} \pqe(\chi)[0/\agvar_1] \equiv \coop{0, \term} \chi[0/\agvar_1]$. The claim now follows from case \emph{1.} in Lemma~\ref{lem:equiv1}. All other cases are proved analogously.
\end{proof}

\begin{lemma}\label{def:pqe-equality}
Let $\varphi$ be any formula in $\langb$, then $\pqe(\varphi) = \varphi$.
\end{lemma}

\begin{proof}
Again, induction on the structure of $\varphi$ in normal form, following the recursive definition of $\pqe$. Note, that the only cases that apply to formulae in normal form are those in lines 5-9, 20 and 22-24, which do not modify $\varphi$.
\end{proof}

Note that after applying $\pqe$, the resulting formula satisfies condition (NF1) in the definition of normal form.

\subsection{Transformation to normal forms and fixpoint equivalences} 
\label{subsec:normalforms2} 

Next, we show that quantification can always be distributed, up to equivalence in the finite, over conjunctions and disjunctions and pushed inside subformulae so that every bound variable is immediately preceded by a quantifier that binds it, which will be used further for transformations of \lang\ formulae to normal form.

  We define by recursion a 2-argument function $\pushf$, applied to pairs consisting of an admissible quantifier prefix $\mathcal{Q}$ and a formula $\varphi$ in $\lang$, such that $\pushf(\mathcal{Q}, \varphi)$ is a formula in $\lang$ which 
satisfies conditions (NF2) and (NF3) of the definition of normal form,
  and which we will prove to be equivalent to  $\mathcal{Q} \varphi$. For the purpose of defining $\pushf$ as described, we will need to define it on a wider scope, viz. applied to any state or path formula $\xi$, even though $\mathcal{Q} \xi$ may not be a legitimate formula of \lang.
    In what follows, we denote by $\ol{\mathcal{Q}}$ the swap of the quantifiers in the prefix $\mathcal{Q}$ with their duals, i.e. $\exists$ with $\forall$ and vice versa.

 \begin{algorithmic}[1]
   \State \textbf{Input}: a quantifier prefix $\mathcal{Q}$ and a (state or path) formula $\xi$ of \lang

   \State \textbf{Output}: $\xi'$, a (state or path) formula of \lang

   \Procedure{$\pushf$}{$\mathcal{Q}, \xi$}

   \Cases{$\xi$}

   \Case{$\top \mid \avarprop$}\label{alg:push-base} \Return $\xi$ 

   \Case{$\neg \psi$}\label{alg:push-neg} \Return $\neg \pushf(\ol{\mathcal{Q}}, \psi)$

   \Case{$\psi_1 \land \psi_2$}\label{alg:push-and} \Return $\pushf(\mathcal{Q}, \psi_1) \land \pushf(\mathcal{Q}, \psi_2)$

   \Case{$\psi_1 \lor \psi_2$}\label{alg:push-or} \Return $\pushf(\mathcal{Q}, \psi_1) \lor \pushf(\mathcal{Q}, \psi_2)$

    \Case{$\coop{\term_1, \term_2} \chi$}\label{alg:push-strat}  
    
       	\Cases{$\mathcal{Q}$}

    	    \Case{$\Q \agvar_i$, where $\term_i = \agvar_i$, for $i=1$ or $i=2$} 
		
\Case{or  
	    $\Q \agvar_1 \Q' \agvar_2$ or $\Q' \agvar_2 \Q \agvar_1$, where $\term_1 = \agvar_1, \term_2 = \agvar_2$}	    
	
	   \hspace{\fill} \Return $\mathcal{Q} \coop{\term_1, \term_2} \pushf(\mathcal{Q}, \chi)$

		   \Case{$\Q \agvar_i$, where $\term_i \neq \agvar_i$, for $i=1$ or $i=2$} 
		   
\Case{or  
	   $\Q \agvar_1 \Q' \agvar_2$ or $\Q' \agvar_2 \Q \agvar_1$, where $\term_1 \neq  \agvar_1, \term_2 \neq  \agvar_2$}
		   
	 \hspace{\fill}   \Return $\coop{\term_1, \term_2} \pushf(\mathcal{Q}, \chi)$

 \Case{$\Q \agvar_1 \Q' \agvar_2$ or $\Q' \agvar_2 \Q \agvar_1$, where $\term_1 = \agvar_1, \term_2 \neq \agvar_2$}
		     
	 \hspace{\fill}     \Return $\Q \agvar_1 \coop{\term_1, \term_2} \pushf(\mathcal{Q}, \chi)$

 \Case{$\Q \agvar_1 \Q' \agvar_2$ or $\Q' \agvar_2 \Q \agvar_1$, where $\term_1 \neq \agvar_1, \term_2 = \agvar_2$}

	 \hspace{\fill}     \Return $\Q' \agvar_2 \coop{\term_1, \term_2} \pushf(\mathcal{Q}, \chi)$     	    
	
           \EndCases

     \Case{$\Q'' \agvar_k \psi$}
 
     \Cases{$\mathcal{Q}$}

    	    \Case{$\Q \agvar_i$, where  $i=k$}  
	   \Return $\pushf(\Q'' \agvar_k, \psi)$ 
	    
    	    \Case{$\Q \agvar_i$, where  $i \neq k$} 
	   \Return $\pushf(\Q \agvar_i\Q'' \agvar_k, \psi)$

    	    \Case{$\Q \agvar_i \Q' \agvar_j$, where  $i = k$} 
	   \Return $\pushf(\Q' \agvar_j\Q'' \agvar_k, \psi)$  

    	    \Case{$\Q \agvar_i \Q' \agvar_j$, where  $j = k$} 
	   \Return $\pushf(\Q \agvar_i\Q'' \agvar_k, \psi)$  
	    	       
           \EndCases

   \Case{$\atlx \psi$} \label{alg:push-atlx} \Return $\atlx \pushf(\mathcal{Q}, \psi)$

   \Case{$\atlg \psi$}\label{alg:push-atlg} \Return $\atlg \pushf(\mathcal{Q}, \psi)$

   \Case{$\psi_1 \atlu \psi_2$}\label{alg:push-atlu} \Return $\pushf(\mathcal{Q}, \psi_1) \atlu \pushf(\mathcal{Q}, \psi_2)$
  \EndCases
   \EndProcedure
 \end{algorithmic}

It is quite easy to see that $\pushf(\mathcal{Q}, \varphi)$ is a formula in $\lang$ whenever 
$\mathcal{Q} \varphi$ is a formula in $\lang$. 
 Intuitively, the function $\pushf$ recursively pushes the quantifier prefix $\mathcal{Q}$ inside the formula by either swapping it when negation occurs or by distributing it over the others boolean connectives until it vanishes. 
 When a strategic operator, possibly with variables as arguments that are quantified by $\mathcal{Q}$ is reached, then $\mathcal{Q}$ is placed in front of the strategic operator and is also distributed in its temporal objective, but the vacuous quantification occurring in the process is removed.
Lastly, when the formula begins with another quantifier $\Q'' \agvar_k$, then it is prefixed by $\mathcal{Q}$, the resulting vacuous quantification, if any, is removed, and the resulting prefix is pushed inside.
   
   \begin{example} \label{ex:push}
     Let $\varphi \in \lang$ be 
     \[ \begin{array}{c}
\forall \agvar_1
 \big( \coop{\agvar_1, 5} (\forall \agvar_2 \coop{\agvar_1, \agvar_2} \atlx p_1) \; \atlu \; (\exists \agvar_2 \coop{\agvar_1, \agvar_2} \atlf p_2 ) \; \lor{} \\
 \exists \agvar_1 (\forall \agvar_2 \coop{\agvar_1, \agvar_2} \atlf p_3  \; \land \; \neg \forall \agvar_2 \coop{3, \agvar_2} \atlx p_1 ) \big)
       \end{array} \]
     where $p_1 , p_2, p_3 \in \Phi$. Then $\pushf(\forall \agvar_1, \varphi) = $ 
   \[ \begin{array}{c}
 \forall \agvar_1 \coop{\agvar_1, 5} \big( (\forall \agvar_1 \forall \agvar_2 \coop{\agvar_1, \agvar_2} \atlx p_1) \;\; \atlu \;\; (\forall \agvar_1 \exists \agvar_2 \coop{\agvar_1, \agvar_2} \atlf p_2 ) \big) \lor{} \\
(\exists \agvar_1 \forall \agvar_2 \coop{\agvar_1, \agvar_2} \atlf p_3  \;\; \land \;\; \neg \forall \agvar_2 \coop{3, \agvar_2} \atlx p_1)
       \end{array} \]    
    \end{example}

\begin{theorem}
\label{thm:pushing-quantifiers}
Let
  $\mathcal{Q}$ be an admissible quantifier prefix
  and let $\mathcal{Q} \varphi$ be a formula of  \lang. 
  Then $\pushf(\mathcal{Q},\varphi)$ is logically equivalent in the finite  to $\mathcal{Q}\varphi$.
\end{theorem}

\begin{proof}
We prove the claim by induction on the \defstyle{nesting depth} $\nd(\varphi)$ of strategic operators in the state formula $\varphi$, defined as expected.

When $\nd(\varphi) = 0$ the claim is straightforward because any quantification over $\varphi$ is vacuous, hence $\mathcal{Q}\varphi \equiv \varphi$ and $\pushf(\mathcal{Q},\varphi) = \varphi$.   
Suppose now that $\nd(\varphi) > 0$ and the claim holds for all state formulae of \lang with lower nesting depth. 
We will do a nested induction on the structure of $\varphi$, following the recursive definition of $\pushf$.
    
  \begin{enumerate}
  \itemsep = 3pt
  \item $\varphi = \top \mid \avarprop$. This case does not apply now, but it is, anyway, trivial for every $\mathcal{Q}$. 
  
  \item $\varphi = \neg \psi$ follows from FOL and the inductive hypothesis (IH) for $\psi$. 
  
  \item $\varphi = \psi_1 \land \psi_2$. 
  
  \begin{enumerate}
   \smallskip 
\item 
  When $\mathcal{Q}=\forall \agvar_i$, the claim follows immediately from the valid equivalence (proved just like in FOL) $\forall \agvar_i (\psi_1 \land \psi_2) \equiv \forall \agvar_i \psi_1 \land \forall \agvar_i \psi_2$ and the IH for each of $\psi_1$ and $\psi_2$.   
  
   \smallskip 
\item 
 When $\mathcal{Q} =\exists \agvar_i$, it suffices to prove that 
 $\exists \agvar_i (\psi_1 \land \psi_2) 
\fineq 
 \exists \agvar_i \psi_1 \land  \exists \agvar_i \psi_2$, 
 and then use the IH for each of $\psi_1$ and $\psi_2$.  
The implication from left to right is by the validity of the implication 
$\exists \agvar_i(\psi_1 \land \psi_2) \to \exists \agvar_i \psi_1 \land \exists \agvar_i \psi_2$. 
To prove the converse implication, first note that, since $\exists \agvar_i (\psi_1 \land \psi_2)$ is a formula of  $\lang$, all free occurrences of $\agvar_i$ in $\psi_1$ and in $\psi_2$ must be positive. 
Now, suppose first that $i=1$ and let 
$\dmas, \sts,  \agass \models  \exists \agvar_1 \psi_1 \land  \exists \agvar_1 \psi_2$  
for some finite $\dmas$. 
Then,  $\dmas, \sts,  \agass \models \psi_1[C_1/ \agvar_1]$  and 
$\dmas, \sts,  \agass \models  \psi_2[C_2/ \agvar_1]$ for some $C_1,C_2 \in \bbN$.  
Let $C =  \max(C_1,C_2)$. By the monotonicity property  \cmon$^+$ from Lemma \ref{lem:monotonicity}, we obtain that $\dmas, \sts,  \agass \models \psi_1[C/ \agvar_1]$  and 
$\dmas, \sts,  \agass \models  \psi_2[C/ \agvar_1]$.
Therefore, $\dmas, \sts,  \agass \models (\psi_1 \land \psi_2)[C/ \agvar_1]$,  hence  
$\dmas, \sts,  \agass \models \exists \agvar_1 (\psi_1 \land \psi_2)$.  
This proves the validity of the converse implication 
$(\exists \agvar_1 \psi_1 \land  \exists \agvar_1 \psi_2) \to \exists \agvar_1 (\psi_1 \land \psi_2)$. 
The proof of the case where $i=2$ is analogous, using the monotonicity property  \nmon$^+$ from Lemma \ref{lem:monotonicity}.  

   \smallskip 
\item 
Lastly, the case when $\mathcal{Q} = \Q\agvar_i\Q'\agvar_j$ is readily reducible to the previous 2 cases, by distributing 
first $\Q'\agvar_j$ and then $\Q\agvar_i$.  
\end{enumerate}

\medskip
 \item $\varphi = \psi_1 \lor \psi_2$. This case is dually analogous to the previous one. 

  \begin{enumerate}
   \smallskip 
\item 
 When $\mathcal{Q} =\exists \agvar_i$, the claim follows immediately from the valid equivalence 
 $\exists \agvar_i(\psi_1 \lor \psi_2) \equiv \exists \agvar_i \psi_1 \lor \exists \agvar_i \psi_2$ 
 and the IH for each of $\psi_1$ and $\psi_2$.   

   \smallskip 
\item 
  When $\mathcal{Q}=\forall \agvar_i$,  it suffices to prove that 
  $\forall \agvar_i (\psi_1 \lor \psi_2) \fineq \forall \agvar_i \psi_1 \lor \forall \agvar_i \psi_2$, 
 and then use the IH for each of $\psi_1$ and $\psi_2$.  
 The implication from right to left 
 $(\forall \agvar_i \psi_1 \lor \forall \agvar_i \psi_2) \to \forall \agvar_i(\psi_1 \lor \psi_2)$ 
 is a validity, proved just like in FOL. For the converse implication, suppose first that $i=1$ and let 
$\dmas, \sts,  \agass \models \forall \agvar_1 (\psi_1 \lor \psi_2)$ for some finite $\dmas$. 
Then, $\dmas, \sts,  \agass \models (\psi_1 \lor \psi_2)[0/\agvar_1]$, hence 
 $\dmas, \sts,  \agass \models \psi_1[0/\agvar_1]$ or 
 $\dmas, \sts,  \agass \models \psi_2[0/\agvar_1]$. Suppose w.l.o.g. the former.  
 Then, by the monotonicity property  \cmon$^+$ from Lemma \ref{lem:monotonicity}, we obtain that  
 $\dmas, \sts,  \agass \models \psi_2[C/\agvar_1]$ for any $C\in \nat$, hence 
 $\dmas, \sts,  \agass \models \forall \agvar_1 \psi_1$, so   
 $\dmas, \sts,  \agass \models \forall \agvar_1 \psi_1 \lor \forall \agvar_1 \psi_2$.

For the case that $i=2$, assuming that $\dmas, \sts,  \agass \models \forall \agvar_2 (\psi_1 \lor \psi_2)$,
it follows that  
at least one of  $\dmas, \sts,  \agass \models \psi_1[N/\agvar_2]$ and 
 $\dmas, \sts,  \agass \models \psi_2[N/\agvar_2]$ holds for infinitely many values of $N \in \nat$. 
Suppose w.l.o.g. the former. Then,  by the monotonicity property  \nmon$^+$ from Lemma \ref{lem:monotonicity}, we obtain that  
 $\dmas, \sts,  \agass \models  {\psi_1}[N/\agvar_2]$ for any $N\in \nat$, hence 
 $\dmas, \sts,  \agass \models \forall \agvar_2 \psi_1$, so   
 $\dmas, \sts,  \agass \models \forall \agvar_2 \psi_1 \lor \forall \agvar_2 \psi_2$.  

   \smallskip 
\item 
Lastly, the case when $\mathcal{Q} = \Q\agvar_i\Q'\agvar_j$ is readily reducible to the previous 2 cases, by distributing 
first $\Q'\agvar_j$ and then $\Q\agvar_i$.  
\end{enumerate}

\medskip
    \item $\varphi= \Q'' \agvar_k \psi$. 
    Again, we consider the subcases depending on $\mathcal{Q}$. 
    
  \begin{enumerate}
\item $\mathcal{Q} = \Q \agvar_i$, where  $i=k$.  

    We are to show that  $\Q \agvar_i \Q'' \agvar_i \psi \fineq  \pushf(\Q'' \agvar_i, \psi)$, which follows from 
    $\Q \agvar_i \Q'' \agvar_i \psi \equiv \Q'' \agvar_i \psi$ and the IH.

\item $\mathcal{Q} = \Q \agvar_i$, where  $i \neq k$.

 We are to show that  $\Q \agvar_i \Q'' \agvar_k \psi \fineq  \pushf(\Q \agvar_i\Q'' \agvar_k, \psi)$, which follows from the 
 IH for $\mathcal{Q} = \Q \agvar_i \Q'' \agvar_k$ and $\psi$. 
  
\item $\mathcal{Q} = \Q \agvar_i \Q' \agvar_j$, where  $i = k$.

 We are to show that  $\Q \agvar_k \Q' \agvar_j \Q'' \agvar_k \psi \fineq  \pushf(\Q' \agvar_j\Q'' \agvar_k, \psi)$, which follows from 
 $\Q \agvar_k \Q' \agvar_j \Q'' \agvar_k \psi \equiv \Q' \agvar_j \Q'' \agvar_k \psi$ and the IH.  
 
\item The case $\mathcal{Q} = \Q \agvar_i \Q' \agvar_j$, where  $j = k$, is analogous.
\end{enumerate}

\medskip 
  \item $\varphi = \coop{\term_1, \term_2}\chi$.

This inductive case -- for both inductions, the external one, on $\nd(\varphi)$, and for the nested one, on the structure of $\varphi$ -- is the most involved case, where the finiteness of the models over which we prove the equivalence is used essentially. 
  There are several subcases, depending on $\mathcal{Q}$ and on the main temporal connective of $\chi$.
  The proof for each case is technical and some cases are longer than others, but they all use a similar approach, that essentially hinges on the finiteness of the model and the monotonicity properties 
 from Lemma \ref{lem:monotonicity}. 
  These will allow us to obtain uniformly large enough values of the quantified variables, beyond which the truth values of all strategic subformulae stabilise, and thus to establish the truth of the non-trivial implications. 
  We will provide a representative selection of proofs for some of the cases and will leave out the rest, which are essentially analogous, though possibly even longer.     

\begin{enumerate}
\item 
$\mathcal{Q} = \Q \agvar_i$, where $\term_i = \agvar_i$, for $i=1$ or $i=2$.  

 We are to show that 
   $\Q \agvar_i \coop{\term_1, \term_2}\chi
   \fineq  
   \Q \agvar_i \coop{\term_1, \term_2} \pushf(\Q \agvar_i, \chi)$, assuming the inductive hypothesis for the main state subformulae of $\chi$.  
   We consider the subcases depending on $\Q$, $i$, and the main temporal connective of $\chi$.  

 \smallskip 
\ \ \ \ \textbf{Case ($\forall \agvar_1 \atlg$)}: 
to prove 
$\forall \agvar_1 \coop{\agvar_1, t} \atlg \psi \fineq  \forall \agvar_1 \coop{\agvar_1, t} \atlg \pushf(\forall \agvar_1, \psi)$.  

By the IH for $\psi$, we have that $\forall \agvar_1 \psi \fineq \pushf(\forall \agvar_1, \psi)$. 
\\
So, it suffices to prove that $\forall \agvar_1 \coop{\agvar_1, t} \atlg \psi \fineq \forall \agvar_1 \coop{\agvar_1, t} \atlg \forall \agvar_1 \psi$.  
By Lemma \ref{lem:equiv1}, 
$\forall \agvar_1 \coop{\agvar_1, t} \atlg \psi \equiv \coop{0, t} \atlg \psi[0/\agvar_1]$ and  
$\forall \agvar_1 \coop{\agvar_1, t} \atlg \forall \agvar_1 \psi \equiv  \coop{0, t} \atlg \forall \agvar_1 \psi$. 

So, we have to prove that $\coop{0, t} \atlg \psi[0/\agvar_1] \equiv \coop{0, t} \atlg \forall \agvar_1 \psi$, which follows immediately, since  
$\forall \agvar_1 \psi \equiv \psi[0/\agvar_1]$, by \cmon$^+$ from Lemma \ref{lem:monotonicity}.

 \smallskip 
\ \ \ \ \textbf{Case ($\exists \agvar_1 \atlg$)}: 
to prove 
$\exists \agvar_1 \coop{\agvar_1, t} \atlg \psi \fineq  \exists \agvar_1 \coop{\agvar_1, t} \atlg \pushf(\exists \agvar_1, \psi)$.  

By the IH for $\psi$, we have that $\exists \agvar_1 \psi \fineq \pushf(\exists \agvar_1, \psi)$. 
So, it suffices to prove that $\exists \agvar_1 \coop{\agvar_1, t} \atlg \psi \fineq \exists \agvar_1 \coop{\agvar_1, t} \atlg \exists \agvar_1 \psi$. 
Since $\models \psi \to \exists \agvar_1 \psi$, we obtain validity of the implication 
$\exists \agvar_1 \coop{\agvar_1, t} \atlg \psi \to \exists \agvar_1 \coop{\agvar_1, t} \atlg \exists \agvar_1 \psi$.

For the converse, suppose 
$\dmas, \sts,  \agass \models \exists \agvar_1 \coop{\agvar_1, t} \atlg \exists \agvar_1 \psi$  
for some finite $\dmas$ with state space $\States$, 
assignment $\agass$ and $\sts \in \States$.   
Fix any $C\in \bbN$ such that 
 $\dmas, \sts, \agass \models \coop{C, t} \atlg \exists \agvar_1 \psi$.   
Since $ \agass$ fixes the values of all terms, we can treat $\exists \agvar_1 \psi$ as a closed formula.
Note that, according to the syntax of $\lang$, all occurrences of $\agvar_1$ in $\psi$ are positive. 
Let $W = \stexten{\exists \agvar_1 \psi}{\agass}$ be its extension in $\dmas$ (which depends on $\agass$) and let $\statw \in W$.  
 Let $f: W \to \bbN$ be a mapping assigning to every $\stu \in W$ a number $f(\stu)$ such that 
 $\dmas, \stu, \agass \models \psi[f(\stu)/\agvar_1]$. 
Now, let\footnote{This is where we use the finiteness of the model.} $f^* := \max_{\stu\in W} f(\stu)$ and $C^* := \max(f^*,C)$. 
Then, by \cmon$^+$ from Lemma \ref{lem:monotonicity}, we obtain that  
 $\dmas, \stu, \agass \models \psi[C^*/\agvar_1]$ for each $\stu \in W$, hence  
$\dmas, \sts, \agass \models \coop{C^*, t} \atlg \psi[C^*/\agvar_1]$.  
Therefore,  $\dmas, \sts, \agass \models \exists \agvar_1 \coop{\agvar_1,t} \atlg \varphi$. 
 Thus,  $\exists \agvar_1 \coop{\agvar_1, t} \atlg \exists \agvar_1 \psi \to \exists \agvar_1 \coop{\agvar_1, t} \atlg \varphi$ is valid in the finite, whence the claim.  

 \smallskip 
\ \ \ \ \textbf{Case ($\forall \agvar_2 \atlg$)}: 
to prove 
$\forall \agvar_2 \coop{t,\agvar_2} \atlg \psi \fineq  \forall \agvar_2 \coop{t,\agvar_2} \atlg \pushf(\forall \agvar_2, \psi)$.  

By the IH for $\psi$, we have that $\forall \agvar_2 \psi \fineq \pushf(\forall \agvar_2, \psi)$. 
So, it suffices to prove that $\forall \agvar_2 \coop{t,\agvar_2} \atlg \psi \fineq \forall \agvar_2 \coop{t,\agvar_2} \atlg \forall \agvar_2 \psi$.  
The implication $\models \forall \agvar_2 \coop{t,\agvar_2} \atlg \forall \agvar_2 \psi \to \forall \agvar_2 \coop{t,\agvar_2} \atlg \psi$ follows from 
$\models \forall \agvar_2 \psi \to \psi$, proved just like in FOL.  
For the converse, suppose 
$\dmas, \sts,  \agass \models \forall \agvar_2 \coop{t,\agvar_2} \atlg \psi$  
for some finite $\dmas$ with state space $\States$, 
assignment $\agass$ and $\sts \in \States$.   
Then, for every $N\in \bbN$, it holds that $\dmas, \sts,  \agass \models \coop{t,N} \atlg \psi[N/\agvar_2]$, i.e., there is an abstract positional joint strategy $\sigma_N$ for $\agass(t)$ many controllable agents, such that $\psi$ is true at every state on every outcome play 
enabled by $\sigma_N$ against $N$ uncontrollable agents. Since there are only finitely many abstract positional joint strategies 
for $\agass(t)$ controllable agents in $\dmas$, there is at least one such joint strategy which works for infinitely many values of $N$, and therefore, by \nmon, it will work for all $N\in \bbN$. Let us fix such strategy $\sigma^{\mathbf{c}}$. 
We will show that $\dmas, \sts,  \agass  \models \forall \agvar_2 \coop{t,\agvar_2} \atlg \forall \agvar_2 \psi$ by proving that, for every $N\in \bbN$,  
if $\sigma^{\mathbf{c}}$ is played by $\agass(t)$ many controllable agents it ensures the truth of $\dmas, \sts,  \agass  \models  \coop{t,N} \atlg \forall \agvar_2 \psi$.  
Suppose this is not the case for some $N'\in \bbN$. Then, there is an abstract positional joint strategy $\sigma^{\mathbf{n}}$ for $N'$ uncontrollable agents that guarantees reaching a state $\statw$ where $\forall \agvar_2 \psi$ fails on the unique play $\pi$ generated by the pair of joint strategies $(\sigma^{\mathbf{c}}, \sigma^{\mathbf{n}})$. Thus, 
$\dmas, \statw, \agass  \not\models \forall \agvar_2 \psi$, i.e.,  $\dmas, \statw, \agass  \models \lnot \forall \agvar_2 \psi$. 
Therefore, $\dmas, \statw, \agass  \models \lnot \psi[N''/\agvar_2]$ for some $N''\in \bbN$.  
Let $N^* := \max(N',N'')$. Then, by \nmon$^-$ from Lemma \ref{lem:monotonicity}, we have that 
$\dmas, \statw, \agass  \models \lnot \psi[N^*/\agvar_2]$. Furthermore, the strategy $\sigma^{\mathbf{n}}$ can be trivially extended to $\sigma^{\mathbf{n}*}$ for $N^*$ uncontrollable agents (by letting the extra $N^* - N'$ uncontrollable agents idle), hence the play $\pi$ is still generated by the resulting pair of joint strategies $(\sigma^{\mathbf{c}}, \sigma^{\mathbf{n}*})$ and the state $\statw$ as above will still be reached on it. 
On the other hand, by the choice of $\sigma^{\mathbf{c}}$, when it is played by the $\agass(t)$ many controllable agents against $N^*$ uncontrollable agents it guarantees maintaining forever the truth of $\psi$, i.e., $\dmas, \sts,  \agass \models \coop{t,N^*} \atlg \psi[N^*/\agvar_2]$. In particular, that implies $\dmas, \statw, \agass  \models \psi[N^*/\agvar_2]$ -- a contradiction. Therefore, the assumption that such $N'$ exists is wrong, whence the claim.

 \smallskip 
\ \ \ \ \textbf{Case ($\exists \agvar_2 \atlg$)}: 
to prove 
$\exists  \agvar_2 \coop{t,\agvar_2} \atlg \psi \fineq  \exists \agvar_2 \coop{t,\agvar_2} \atlg \pushf(\exists \agvar_2, \psi)$.  

This case is quite analogous to \textbf{Case ($\forall \agvar_1 \atlg$)} and is proved by using the IH for $\psi$, the equivalences  
$\exists  \agvar_2 \coop{t,\agvar_2} \atlg \psi \equiv \coop{t,0} \atlg \psi[0/\agvar_2]$ and 
$\exists \agvar_2 \coop{t,\agvar_2} \atlg \exists \agvar_2 \psi \equiv \coop{t,0} \atlg \exists \agvar_2 \psi$ from Lemma \ref{lem:equiv1}, and     
the monotonicity properties  $\cmon$ from Lemma \ref{lem:monotonicity}.

 \smallskip 
\ \ \ \  \textbf{Cases} $(\Q \agvar_i \atlx)$ are analogous, but a little simpler than those above.  

\ \ \ \  \textbf{Cases $(\Q \agvar_i \atlu)$}  are analogous, though a little longer than those above.   
    
 \smallskip      
  \item
$\mathcal{Q} = \Q \agvar_i$, where  $\term_1 \not = \agvar_i$ and $\term_2 \not = \agvar_i$.
  
    We are to show that 
   $\Q \agvar_i \coop{\term_1, \term_2}\chi \fineq  \coop{\term_1, \term_2} \pushf(\Q \agvar_i, \chi)$, 
   assuming the IH for the main state subformulae of $\chi$.  For that, it suffices to prove that the quantifier 
   $\Q$ can be equivalently pushed inside through $\coop{\term_1, \term_2}$ and the main temporal connective of $\chi$, e.g.,  
   that $\Q \agvar_i \coop{\term_1, \term_2}\atlg \psi \fineq  \coop{\term_1, \term_2} \Q \agvar_i \atlg \psi$. 
   The non-trivial implications follow from the fact that there are only finitely many abstract positional strategies for the controllable agents 
   in any given finite model, plus the monotonicity properties from Lemma \ref{lem:monotonicity}. The argument for that is essentially the same as that in the proof of \textbf{Case ($\forall \agvar_2 \atlg$)} above.

 \smallskip      
  \item 
$\mathcal{Q} = \Q \agvar_1 \Q' \agvar_2$, where  $\term_1 \neq \agvar_1$ and $\term_2 \neq \agvar_2$.
 
    We are to show that 
   $\Q \agvar_1 \Q' \agvar_2 \coop{\term_1, \term_2}\chi \fineq 
\coop{\term_1, \term_2} \pushf(\Q \agvar_1 \Q' \agvar_2, \chi)$, \\ 
   assuming the IH for all state formulae of lower nesting depth, including the main state subformulae of $\chi$. 
 This equivalence follows by applying case (b) twice, first for $\mathcal{Q} = \Q' \agvar_2 $ and then for 
 $\mathcal{Q} = \Q \agvar_1$ (the IH on the nesting of strategic operators is used here), and each time using the IH.
   
 \smallskip   
 The case  
$\mathcal{Q} = \Q' \agvar_2 \Q \agvar_1$, where  $\term_1 \neq \agvar_1$ and $\term_2 \neq \agvar_2$  is completely analogous. 
   
 \smallskip      
  \item 
$\mathcal{Q} = \Q \agvar_1 \Q' \agvar_2$, where  $\term_1 = \agvar_1$ and $\term_2 \neq \agvar_2$.
    We are to show that  \\
   $\Q \agvar_1 \Q' \agvar_2 \coop{\term_1, \term_2}\chi \fineq 
   \Q \agvar_1 \coop{\term_1, \term_2} \pushf(\Q \agvar_1 \Q' \agvar_2, \chi)$, \\ 
   assuming the IH for all state formulae of lower nesting depth, incl. the main state subformulae of $\chi$. 
   E.g., when $\chi = \atlg \psi$,  we are to prove 
 $\Q \agvar_1 \Q' \agvar_2 \coop{\agvar_1, t_2} \atlg \psi \fineq  
\Q \agvar_1 \coop{\agvar_1, t_2} \atlg \pushf(\Q \agvar_1 \Q' \agvar_2, \psi)$.  
By the IH, $\pushf(\Q \agvar_1 \Q' \agvar_2, \psi) \fineq \Q \agvar_1 \Q' \agvar_2 \psi $, 
so we are to prove that 
$\Q \agvar_1 \Q' \agvar_2 \coop{\agvar_1, t_2} \atlg \psi \fineq  
\Q \agvar_1 \coop{\agvar_1, t_2} \atlg \Q \agvar_1 \Q' \agvar_2 \psi$.   

This follows by first applying case (b) for $\mathcal{Q} = \Q' \agvar_2 $ and the IH to obtain

$\Q' \agvar_2 \coop{\agvar_1, t_2} \atlg \psi \fineq  
\coop{\agvar_1, t_2} \atlg  \pushf(\Q' \agvar_2, \psi) \fineq  
\coop{\agvar_1, t_2} \atlg \Q' \agvar_2 \psi$, and then applying $\Q \agvar_1$ to both sides, 
then case (a) for $\mathcal{Q} = \Q \agvar_1$, and again the IH. 

 \smallskip      
  \item The case 
$\mathcal{Q} =  \Q' \agvar_2 \Q \agvar_1$, where  $\term_1 = \agvar_1$ and $\term_2 \neq \agvar_2$ is similar.
 
 \smallskip      
  \item The cases  
 $\mathcal{Q} = \Q \agvar_1 \Q' \agvar_2$ and  
$\mathcal{Q} =  \Q' \agvar_2 \Q \agvar_1$ 
where  $\term_1 \neq \agvar_1$ and $\term_2 = \agvar_2$ are completely analogous. 

  \smallskip      
  \item 
$\mathcal{Q} = \Q \agvar_1 \Q' \agvar_2$, where  $\term_1 = \agvar_1$ and $\term_2 = \agvar_2$.

    We have to prove 
   $\Q \agvar_1 \Q' \agvar_2 \coop{\agvar_1, \agvar_2}\chi \fineq 
   \Q \agvar_1 \Q' \agvar_2 \coop{\agvar_1, \agvar_2} \pushf(\Q \agvar_1 \Q' \agvar_2, \chi)$, \\ 
   assuming the IH for all state formulae of lower nesting depth, incl. the main state subformulae of $\chi$. 
   E.g., when $\chi = \atlg \psi$,  we are to prove 
$\Q \agvar_1 \Q' \agvar_2 \coop{\agvar_1, \agvar_2} \atlg \psi \fineq  
\Q \agvar_1 \Q' \agvar_2 \coop{\agvar_1, \agvar_2} \atlg \pushf(\Q \agvar_1 \Q' \agvar_2, \psi)$.  

By the IH, $\pushf(\Q \agvar_1 \Q' \agvar_2, \psi) \fineq \Q \agvar_1 \Q' \agvar_2 \psi $, 
so we are to prove that \\ 
$\Q \agvar_1 \Q' \agvar_2 \coop{\agvar_1, \agvar_2} \atlg \psi \fineq  
\Q \agvar_1 \Q' \agvar_2 \coop{\agvar_1, \agvar_2} \atlg \Q \agvar_1 \Q' \agvar_2 \psi$.   

By case (a), we have already shown that \\ 
$\Q' \agvar_2 \coop{\agvar_1, \agvar_2} \atlg \psi \fineq  
 \Q' \agvar_2 \coop{\agvar_1, \agvar_2} \atlg \Q' \agvar_2 \psi$.  
 
 By applying $\Q \agvar_1$ to both sides we obtain \\ 
 $\Q \agvar_1\Q' \agvar_2 \coop{\agvar_1, \agvar_2} \atlg \psi \fineq  
 \Q \agvar_1\Q' \agvar_2 \coop{\agvar_1, \agvar_2} \atlg \Q' \agvar_2 \psi$, so it remains to prove \\   
 $\Q \agvar_1\Q' \agvar_2 \coop{\agvar_1, \agvar_2} \atlg \Q' \agvar_2 \psi
 \fineq 
 \Q \agvar_1 \Q' \agvar_2 \coop{\agvar_1, \agvar_2} \atlg \Q \agvar_1 \Q' \agvar_2 \psi$.  
  
  For each case of $ \Q$ the argument for the non-trivial implication uses the monotonicity properties from Lemma \ref{lem:monotonicity} and is respectively similar to that in the proof of \textbf{Case ($\exists \agvar_2 \atlg$)} and \textbf{Case ($\forall \agvar_2 \atlg$)} above. 

 \smallskip 
The other cases for $\chi$ are similar. 

 \smallskip      
  \item 
The case $\mathcal{Q} = \Q' \agvar_2 \Q \agvar_1$, where  $\term_1 = \agvar_1$ and $\term_2 = \agvar_2$  
is completely analogous to the previous one. 
\end{enumerate}

 \smallskip  
 This completes the proof for all cases in the definition of $\pushf(\mathcal{Q}, \coop{\term_1, \term_2} \chi)$ and, therefore, the last inductive case in both inductions. Q.E.D. 
     \end{enumerate}
     
\end{proof}

Now we will define a recursive function \nff that transforms any state or path formula $\xi$ of $\lang$ respectively into a state or path formula $\xi^{\mathsf{NF}}$ in \langb, while preserving equivalence in the finite.

\bigskip
\begin{algorithmic}[1]
   \State \textbf{Input}: $\xi$, a state or path formula of \lang

   \State \textbf{Output}: $\xi^{\mathsf{NF}}$,  a state or path formula of \langb

   \Procedure{\nff}{$\xi$}

   \Cases{$\xi$}

   \Case{$\top \mid \avarprop$}\label{alg:nf-base} \Return $\xi$ 

   \Case{$\neg \psi$}\label{alg:nf-neg}
   \Return $\neg \nff(\psi)$
   
   \Case{$\psi_1 \land \psi_2$}\label{alg:nf-and} \Return $\nff(\psi_1) \land \nff(\psi_2)$

   \Case{$\psi_1 \lor \psi_2$}\label{alg:nf-or} \Return $\nff(\psi_1) \lor \nff(\psi_2)$

   \Case{$\coop{\term_1, \term_2} \chi$}\label{alg:nf-strat}
   \Return $\coop{\term_1, \term_2} \nff(\chi)$

\Case{$\mathcal{Q} \psi$} \Return{$\pqe(\pushf(\mathcal{Q}, \nff(\psi)))$}

   \Case{$\atlx \psi$} \label{alg:nf-atlx} \Return $\atlx \nff(\psi)$

   \Case{$\atlg \psi$}\label{alg:nf-atlg} \Return $\atlg \nff(\psi)$

   \Case{$\psi_1 \atlu \psi_2$}\label{alg:nf-atlu} \Return $\nff(\psi_1) \atlu \nff(\psi_2)$
   
  \EndCases
   \EndProcedure
 \end{algorithmic}
Intuitively, $\nff$ transforms the input formula by first applying $\pushf$ and then $\pqe$ whenever a quantifier prefix is to be applied, thus producing a formula in a normal form. 

 \begin{example}
   Let $\varphi$ be as in Example~\ref{ex:push}. Then $\nff(\varphi)$ is:
\[ \begin{array}{c}
\coop{0, 5} \big( (\forall \agvar_2 \coop{0, \agvar_2} \atlx p_1) \;\; \atlu \;\;  (\forall \agvar_2 \coop{0, \agvar_2} \atlf p_2 ) \big) \lor{} \\
( \exists \agvar_1\forall\agvar_2 \coop{\agvar_1, \agvar_2} \atlf p_3  \;\; \land \;\; \neg \forall \agvar_2 \coop{3, \agvar_2} \atlx p_1 )
       \end{array} \]
   \end{example}

   \begin{lemma} \label{lem:push}
     Let $\varphi$ a state
     formula of $\langb$. For every admissible quantifier prefix $\mathcal{Q}$, if the variables occurring in $\mathcal{Q}$ do not occur free in $\varphi$, then $\pushf(\mathcal{Q}, \varphi) = \varphi$.
   \end{lemma}
   \begin{proof}
     The argument is by structural induction on $\varphi$ in normal form, by following the recursive definition of $\pushf$. The non-trivial cases are those involving quantifiers. We consider $\exists \agvar_1\coop{\agvar_1, \term_2} \chi$ and $\mathcal{Q} = \Q\agvar_1\Q'\agvar_2$, the other cases are proved analogously.
Since $\varphi \in \langb$ by hypothesis, we have that $\term_2 \not = \agvar_2$, thus by definition
$\pushf(\Q\agvar_1\Q'\agvar_2, \exists\agvar_1 \coop{\agvar_1, \term_2} \chi) =$
    $\pushf(\Q'\agvar_2 \exists \agvar_1, \coop{\agvar_1, \term_2} \chi) =$
    $\exists \agvar_1 \coop{\agvar_1, \term_2} \pushf(\Q'\agvar_2 \exists \agvar_1, \chi)$. By hypothesis $\agvar_2$ is not free in $\psi$, which entails that $\agvar_2$ is not free in $\chi$ and the same holds for $\agvar_1$ by (NF2) in the definition of normal form. We can therefore apply the inductive hypothesis on $\chi$ to get $\exists \agvar_1 \coop{\agvar_1, \term_2} \pushf(\Q'\agvar_2 \exists \agvar_1, \chi) = \exists \agvar_1 \coop{\agvar_1, \term_2} \chi$.
     \end{proof}

   \begin{lemma} \label{lem:nff-closure}
If $\varphi \in \langb$ then $\nff(\varphi)=\varphi$.
\end{lemma}
\begin{proof}
  By induction on the structure of $\varphi$ in normal form, following the recursive definition of $\nff$. 
 
  All cases are straightforward, except $\varphi = \mathcal{Q} \psi$. 
  We consider $\varphi = \exists \agvar_1 \coop{\agvar_1, \term_2} \chi$, all other cases being analogous. 
By the IH, $\nff(\coop{\agvar_1, \term_2} \chi) = \coop{\agvar_1, \term_2} \chi$.  
Also, note that $\agvar_1$ does not occur free in $\chi$ since $\varphi \in \langb$. 
Therefore, using Lemmas \ref{lem:push} and  \ref{def:pqe-equality}, we successively obtain:
$\nff(\varphi) = 
\nff(\exists \agvar_1 \coop{\agvar_1, \term_2} \chi) = 
  \pqe(\pushf(\exists \agvar_1, \nff(\coop{\agvar_1, \term_2} \chi))) =  \\
  \pqe(\pushf(\exists \agvar_1, \coop{\agvar_1, \term_2} \chi)) =   
   \pqe(\exists \agvar_1 \coop{\agvar_1, \term_2} \chi) =
   \exists \agvar_1 \coop{\agvar_1, \term_2} \chi 
    = \varphi $.
  \end{proof}
 
  \begin{theorem} 
\label{thm:langb} 
Let $\varphi$ be any formula in $\lang$. Then: 
\begin{enumerate} 
\item $\nff(\varphi) \in \langb$. 

\item $\nff(\varphi) \fineq \varphi$.

\item $\nff(\varphi)$ can be computed effectively and has length linearly bounded above by $|\varphi|$. 
\end{enumerate}
\end{theorem}

\begin{proof}
  The first claim follows by straightforward induction on the structure of $\varphi$, or just by direct inspection of the function $\nff$.
  
  Claim \emph{2.} is proved by induction on the structure of $\varphi$, following the cases of the recursive definition of $\nff$. The only non-trivial case is  $\varphi = \mathcal{Q} \psi$, which follows immediately from the IH, Theorem~\ref{thm:pushing-quantifiers}, and Lemma~\ref{def:pqe-equality}.
  
     Lastly, Claim \emph{3.} follows by direct inspection of all cases in the definitions of the functions $\pqe$, $\pushf$ and $\nff$.
  \end{proof}

We conclude the section by presenting the fixpoint characterizations of formulae in~(\ref{eq:nf-temporal}), which provide an effective procedure for the model checking algorithm.

\begin{theorem} 
\label{lem:fixpoints1}
For every terms $t, t',t'' \in \termset$ 
the following equivalences hold, where the formulae on the left are in \langb. 
\begin{enumerate}
\smallskip
\item 
$\coop{t', t''} \atlg \varphi \equiv  
\varphi \land \coop{t', t''} \atlx \coop{t', t''} \atlg \varphi $  

\smallskip
\item 
$\coop{t', t''} \psi \atlu \varphi \equiv  
\varphi \lor (\psi \land \coop{t', t''} \atlx \coop{t', t''} \psi \atlu \varphi)$  

\smallskip
\item 
$\exists y_1 \coop{y_1, t} \atlg \varphi \fineq 
\varphi \land \exists y_1 \coop{y_1,t} \atlx \exists y_1 \coop{y_1,t} \atlg \varphi $  

\smallskip
\item 
$\forall y_2 \coop{t, y_2} \atlg \varphi \fineq 
\varphi \land \forall y_2 \coop{t, y_2} \atlx \forall y_2  \coop{t, y_2} \atlg \varphi $  

\smallskip
\item 
$\exists y_1 \coop{y_1, t} \psi \atlu \varphi  \fineq 
 \varphi \lor (\psi \land \exists y_1 \coop{y_1, t} \atlx \exists y_1 \coop{y_1, t} \psi \atlu \varphi)$  

\smallskip
\item 
$\forall y_2 \coop{t, y_2} \psi \atlu \varphi \fineq 
\varphi \lor (\psi \land  \forall y_2 \coop{t, y_2} \atlx \forall y_2 \coop{t, y_2} \psi \atlu \varphi)$  

\smallskip
\item 
$\forall y_2 \exists y_1 \coop{y_1, y_2} \atlg \varphi \fineq  
\varphi \land \forall y_2 \exists y_1 \coop{y_1, y_2} \atlx \forall y_2 \exists y_1 \coop{y_1, y_2} \atlg \varphi$.     

\smallskip
\item 
$\exists y_1 \forall y_2 \coop{y_1, y_2} \atlg \varphi \fineq  \varphi \land  \exists y_1 \forall y_2 \coop{y_1, y_2} \atlx 
 \exists y_1 \forall y_2 \coop{y_1, y_2} \atlg \varphi$.    
  
  \smallskip
 \item 
$\forall y_2 \exists y_1 \coop{y_1, y_2} \psi \atlu \varphi \fineq 
 \varphi \lor ( \psi \land  \forall y_2 \exists y_1 \coop{y_1, y_2} \atlx 
\forall y_2 \exists y_1 \coop{y_1, y_2} \psi \atlu \varphi)$.     

\smallskip
\item 
$\exists y_1 \forall y_2 \coop{y_1, y_2} \psi \atlu \varphi \fineq 
 \varphi \lor  (\psi \land  \exists y_1 \forall y_2 \coop{y_1, y_2} \atlx 
 \exists y_1 \forall y_2 \coop{y_1, y_2} \psi \atlu \varphi)$.    
 \end{enumerate}
\end{theorem}

\begin{proof} 
~

\begin{enumerate}
\item 
Follows directly from the semantics, just like the respective fixpoint equivalence for 
$\coop{A} \atlg$\! in\! ATL, cf.\! \cite{TLCSbook}. 

\item 
Likewise, just like the respective fixpoint equivalence for 
$\coop{A}\! \atlu$ in ATL.

\item 
First, note that $y_1$ does not occur free in $\varphi$ since $\exists y_1 \coop{y_1, t} \atlg \varphi \in \langb$. 

Now, we take the equivalence 1, where $t' = y_1$, and quantify both sides with $\exists y_1$, obtaining:

$\exists y_1 \coop{y_1, t} \atlg \varphi \equiv 
\exists y_1 (\varphi \land \coop{y_1, t} \atlx \coop{y_1, t} \atlg \varphi) 
\equiv 
\varphi \land \exists y_1 \coop{y_1, t} \atlx \coop{y_1, t} \atlg \varphi$.       
  
  By applying Theorem~\ref{thm:langb} to both sides  above
  and then using Lemmas \ref{lem:push},  \ref{lem:nff-closure} and \ref{def:pqe-equality}
we obtain:  
  
  $\exists y_1 \coop{y_1, t} \atlg \varphi   
  \\ 
  \fineq     
 \nff(\varphi \land \exists y_1\coop{y_1, t} \atlx \coop{y_1, t} \atlg \varphi) \\ 
 =  \nff(\varphi)  \land \nff(\exists y_1\coop{y_1, t} \atlx \coop{y_1, t} \atlg \varphi) \\   
=  \varphi  \land   \pqe(\pushf(\exists y_1, \nff(\coop{y_1, t} \atlx \coop{y_1, t} \atlg \varphi))) \\ 
=  \varphi  \land   \pqe(\pushf(\exists y_1, \coop{y_1, t} \atlx \coop{y_1, t} \atlg \nff(\varphi))) \\ 
=  \varphi  \land   \pqe(\pushf(\exists y_1, \coop{y_1, t} \atlx \coop{y_1, t} \atlg \varphi)) \\ 
=  \varphi  \land   \pqe(\exists y_1 \coop{y_1, t}  \pushf(\exists y_1, \atlx \coop{y_1, t} \atlg \varphi)) 
\\ 
=  \varphi  \land   \pqe(\exists y_1 \coop{y_1, t}  \atlx \exists y_1\coop{y_1, t} \pushf(\exists y_1, \atlg \varphi))  \\ 
=  \varphi  \land   \pqe(\exists y_1 \coop{y_1, t}  \atlx \exists y_1\coop{y_1, t} \atlg \varphi)  \\ 
\equiv 
\varphi \land \exists y_1 \coop{y_1, t} \atlx \exists \agvar_1 \coop{\agvar_1, t} \atlg \varphi$.

\medskip
The other cases are analogous.
\end{enumerate}
\end{proof}


%% file: modelChecking.tex
\section{Model checking}
\label{sec:MC}

In this section we develop an algorithm for model checking the fragment \langb.  
By virtue of Theorem \ref{thm:langb}, it will provide a model checking procedure for the whole 
\lang. 

Let $\varphi$ be any  state formula of \lang, $\dmas$ be a \oursys, $\sts$ a state and $\agass$ an assignment in $\dmas$. The \defstyle{local model checking problem} is the problem of deciding whether $\dmas, \sts, \agass \models \varphi$, while 
the \defstyle{global model checking problem} is the computational problem that returns the set of 
 states in $\dmas$ where the input formula $\varphi$ is satisfied, i.e. it is the problem of computing the  \defstyle{state extension} of $\varphi$ in $\dmas$ given $\agass$, formally defined as:
\[\stexten{\varphi}{\agass} = \set{\sts \in \States \mid \dmas, \sts, \agass \models \varphi}.\]
For closed formulae $\varphi$, $\stexten{\varphi}{\agass}$ does not depend on the assignment $\agass$, so we omit it and write 
$\stexten{\varphi}{}$.

Algorithm~\ref{algo:glob-mc} presented here solves the global model checking problem for all $\langb$ formulae.
The core sub-procedure of the algorithm is the function $\algostyle{preImg}$ which, given a set $\Qstates$ of states in $\States$ and $C, N \in \nat$, returns the set of states from which $C$ controllable agents have a joint action, which, when played against any joint action of other $N$ uncontrollable agents produces an outcome state in $\Qstates$. 
We will call that set the \defstyle{$(C,N)$-controllable pre-image of $\Qstates$}. Often we will omit $(C,N)$, when unspecified or fixed in the context, and will write simply ``the controllable pre-image of $\Qstates$''.   We also extend that notion to `` \defstyle{$(\term_1,\term_2)$-controllable pre-image}'', for any
terms $\term_1,\term_2$, 
the values of which are given by the assignment.  When $\Qstates = \stexten{\psi}{{\agass}}$, it computes the state extension of $\coop{\term_1,\term_2} \atlx \fob$ which is parameterised by terms $\term_1,\term_2$ (by means of their values 
$\agass(\term_1)$ and $\agass(\term_2)$).
We then extend that further to quantified extensions of $\coop{\term_1,\term_2} \atlx \fob$, by adding the respective quantification to the result.   
In all cases, we reduce the problem of computing the controllable pre-images to checking the truth of  Presburger formulae.

We now proceed with some technical preparation. Recall that $\actvarset^+$ is the set of $n+1$ action counters.
We will also be using auxiliary integer variables $k_1, \ldots, k_n, k_{\noop}$ and $\ell_1, \ldots, \ell_n, \ell_{\noop}$ not contained in $\actvarset^+$. Each $k_i$ (respectively, $\ell_i$) represents the number of {controllable} (respectively, {uncontrollable}) agents performing action $\act_i$; likewise for $k_{\noop}$ (resp., $\ell_{\noop}$) for the number of controllable (resp., uncontrollable) agents performing the idle action.
Also, for each $\sts$ in $\States$ and $i \in \set{1, \ldots, n}$ we introduce an auxiliary propositional constant $\actav^i_s$ which is true
if and only if action $\act_i$ is available in $\sts$, i.e., $\act_i \in \actav(\sts)$.

\begin{definition}\label{def:Presb}
Given a \oursys $\dmas$ with a finite  state space $\States$, 
state $\sts$ in $\States$, a subset $\Qstates$ of $\States$, and 
terms $\term_1$, $\term_2$, 
we define the following Presburger formulae: 

\[ \grd{\sts}{\Qstates}(\actvar_1,..., \actvar_n) := \bigvee_{\sts' \in \Qstates} \transrel(\sts, \sts')(\actvar_1,..., \actvar_n).  \hspace{24mm}  \]  
 
 \smallskip  
\[ \prf(\dmas, \sts, \term_1, \term_2, \Qstates) := 
 \exists k_1 ... \exists k_n\, \exists k_{\noop}  
 \Bigg( 
 \bigwedge_{i=1}^n (k_i \not = 0 \ra \actav^i_s) \land{} 
  \sum_{i =1}^n k_i + k_{\noop} = \term_1 \land {} 
  \hspace{\fill} ~ 
          \]
        \[  
 \forall \ell_1...  \forall \ell_n\,  \forall \ell_{\noop} \,        
       \bigg(   
       \Big(\bigwedge_{i=1}^n (\ell_i \not = 0 \ra \actav^i_s ) \land{} 
           \sum_{i =1}^n \ell_i + \ell_{\noop} = \term_2 
          \Big)   
    \ra 
     \grd{\sts}{\Qstates}\big((k_1+\ell_1),...,(k_n+\ell_n)\big)     
       \bigg) 
     \Bigg) 
    \]
\end{definition}

The formula $\prf(\dmas, \sts, \term_1, \term_2, \Qstates)$ intuitively says that there is a tuple of available actions at $\sts$ such that when played by $ \term_1$ many (controllable) agents and combined with any tuple of available actions for $ \term_2$ many (uncontrollable) agents, it satisfies a guard of a transition leading to a state in $Q$. 
(The formula can be shortened somewhat, if the quantification is restricted only to $k-$ and $\ell-$variables that correspond to action counters that appear in the guard $\grd{\sts}{\Qstates}$, which  would improve the complexity estimates, as shown in Section \ref{sec:complexity}.)    
That formula and its extensions with quantifiers over $\term_1$ (when equal to $\agvar_1$) and $\term_2$ (when equal to $\agvar_2$) will be used by the global model checking algorithm to compute the controllable pre-images of state extensions.

\begin{example}
Let us compute the state extension of the formula \\ $\varphi = \exists \agvar_1 \forall \agvar_2 \coop{\agvar_1, \agvar_2} \atlx (p \lor q)$ in the model $\dmas$ of Example~\ref{ex:1}. First, we compute  $\stexten{p \lor q}{} = \set{\sts_2, \sts_3, \sts_4, \sts_5, \sts_6}$. Then, for each state $\sts \in \dmas$ we check the truth of the closed Presburger formula 
 $\exists \agvar_1 \forall \agvar_2 \prf(\dmas, \sts, \agvar_1, \agvar_2, \stexten{p \lor q}{})$ in $\dmas$. 
  \begin{itemize}
  \item $\exists \agvar_1 \forall \agvar_2 \, \prf(\dmas, \sts_1, \agvar_1, \agvar_2, \stexten{p \lor q}{})$ is false, thus $\sts_1$ does not belong to the $\exists y_1\forall y_2(y_1,y_2)$-controllable pre-image of $\stexten{p \lor q}{}$. Indeed $11$ uncontrollable agents can force the system to stay in $\sts_1$ when they all perform $\act_3$;
  \item $\exists \agvar_1 \forall \agvar_2  \, \prf(\dmas, \sts_2, \agvar_1, \agvar_2, \stexten{p \lor q}{})$ is true, hence $\sts_2$ belongs to the $\exists y_1\forall y_2(y_1,y_2)$-controllable pre-image of $\stexten{p \lor q}{}$ trivially because all outgoing transitions from $\sts_2$ lead to states in $\stexten{p \lor q}{}$;
\item checking all other states likewise produces the final result: \\ 
$\stexten{\varphi}{} = \set{\sts_2, \sts_4, \sts_5, \sts_6}$.
\end{itemize}
  \end{example}


\begin{algorithm}
\caption{Computing the controllable by \! $\term_1$ agents pre-image of $\Qstates$ against $\term_2$ agents (with $\term_1$, $\term_2$ possibly quantified).}
 \label{algo:pre-img}
 \begin{algorithmic}[1]
   \State \textbf{Inputs}: \oursys $\dmas$, $\term_1, \term_2, \in \termset$, $\Qstates \subseteq \States$, assignment $\agass$ and prefix $\pfix$
   
   \State \textbf{Output}: the ($\term_1$, $\term_2$)-controllable pre-image $Z \subseteq \States$ of $\Qstates$
   \Procedure{preImg}{$\dmas, \term_1, \term_2, \Qstates, \agass, \pfix$}

   \If{$\term_1 \neq y_1$ does not appear in $\pfix$}
   \State $\term_1 \gets \agass(\term_1)$
\EndIf
   
      \If{$\term_2 \neq y_2$  does not appear in $\pfix$}
   \State $\term_2 \gets \agass(\term_2)$
\EndIf
    
\State $Z \gets \emptyset$
   \ForAll{$\sts \in \States$}
   \If{$\pfix \, \prf(\dmas, \sts, \term_1, \term_2, \Qstates)$ true} $Z \gets Z \cup \set{\sts}$
   \EndIf
   \EndFor
   \State \Return $Z$
\EndProcedure
    \end{algorithmic}
  \end{algorithm}


  \begin{algorithm}
\caption{Global model checking algorithm for closed formulae of the type $\pfix  \coop{\term_1, \term_2} \atlg \fob$.}
 \label{algo:G-fixpoint}
 \begin{algorithmic}[1]
   \State \textbf{Inputs}: \oursys $\dmas$, $\term_1, \term_2 \in \termset$, formula $\fob$, assignment $\agass$ and prefix $\pfix$
   \State \textbf{Output}: the set of states $\aset = \stexten{\pfix \coop{\term_1, \term_2} \atlg \fob}{\agass}$
   
   \Procedure{G-fixpoint}{$\dmas, \term_1, \term_2, \fob, \agass, \pfix$}

   \State $\Qstates \gets \algostyle{globalMC}(\dmas, \fob, \agass)$

   \State $\asetbis \gets \States$
   \State $\aset \gets \Qstates$
\While{$\asetbis \not\subseteq \aset$}    
\State $\asetbis \gets \aset$
\State $\aset \gets \algostyle{preImg}(\dmas, \term_1, \term_2, \pfix, \asetbis) \cap \Qstates$ 
\EndWhile
\State \Return $\aset$

\EndProcedure
    \end{algorithmic}
  \end{algorithm}



  \begin{algorithm}
    \caption{Global model checking algorithm for closed formulae of the type $\pfix  \coop{\term_1, \term_2} \fob_1 \atlu \fob_2$.}
 \label{algo:U-fixpoint}
 \begin{algorithmic}[1]
   \State \textbf{Inputs}: \oursys $\dmas$, $\term_1, \term_2 \in \termset$, formulae $\fob_1, \fob_2$, assignment $\agass$ and prefix $\pfix$
   \State \textbf{Output}: the set of states $\aset = \stexten{\pfix  \coop{\term_1, \term_2} \fob_1 \atlu \fob_2}{\agass}$
   \Procedure{U-fixpoint}{$\dmas, \term_1, \term_2, \fob_1, \fob_2, \agass, \pfix$}

   \State $\Qstates_1 \gets \algostyle{globalMC}(\dmas, \fob_1, \agass)$
    \State $\Qstates_2 \gets \algostyle{globalMC}(\dmas, \fob_2, \agass)$

    \State $\asetbis \gets \emptyset$
    \State $\aset \gets  \Qstates_2$   
\While{$\aset \not\subseteq \asetbis$}    
\State $\asetbis \gets \aset$
\State $\aset \gets \Qstates_2 \cup (\algostyle{preImg}(\dmas,\term_1, \term_2, \pfix, \asetbis) \cap 		
	\Qstates_1)$ 
        \EndWhile
        \State \Return $\aset$

\EndProcedure
    \end{algorithmic}
  \end{algorithm}


\begin{algorithm}
\caption{Global model checking algorithm for $\lang$-formulae.}
 \label{algo:glob-mc}
 \begin{algorithmic}[1]
\State \textbf{Inputs}: \oursys $\dmas$, formula $\varphi$ and assignment $\agass$
   \State \textbf{Output}: the set of states $\aset = \stexten{\varphi}{\agass}$

\Procedure{\algostyle{globalMC}}{$\dmas, \fo, \agass$}
\Cases{$\fo$}

\Case{$\avarprop$}
\State \Return $\{\astate \in \States \mid \avarprop \in \labf(\astate)\}$

\Case{$\lnot \fob$}
\State \Return $\States \setminus \algostyle{globalMC}(\dmas, \fob, \agass)$

\Case{$\fob_{1} \land \fob_{2}$}
\State $\Qstates_1 \gets \algostyle{globalMC}(\dmas, \fob_1, \agass)$
\State $\Qstates_2 \gets \algostyle{globalMC}(\dmas, \fob_2, \agass)$
\State \Return $\Qstates_1 \cap  \Qstates_2$

\Case{$\fob_{1} \lor \fob_{2}$}
\State $\Qstates_1 \gets \algostyle{globalMC}(\dmas, \fob_1, \agass)$
\State $\Qstates_2 \gets \algostyle{globalMC}(\dmas, \fob_2, \agass)$
\State \Return $\Qstates_1 \cup  \Qstates_2$

\Case{$ \pfix  \coop{\term_1, \term_2}\atlx \fob$}
\State $\Qstates \gets \algostyle{globalMC}(\dmas, \fob, \agass)$
\State \Return $\algostyle{preImg}(\dmas, \term_1, \term_2, \Qstates, \agass, \pfix)$

\Case{$\pfix  \coop{\term_1, \term_2}\atlg \fob$}
\State \Return $\algostyle{G-fixpoint}(\dmas, \term_1, \term_2, \fob, \agass, \pfix)$

\Case{$ \pfix  \coop{\term_1, \term_2} \fob_{1} \atlu \fob_{2}$}   
\State \Return $\algostyle{U-fixpoint}(\dmas, \term_1, \term_2, \fob_1, \fob_2, \agass, \pfix)$
\EndCases
\EndProcedure
    \end{algorithmic}
  \end{algorithm}

 We now present the global model checking Algorithm~\ref{algo:glob-mc}.
 From here on, we denote by $\pfix$ any string from the set
  $\set{\emptystring, \exists \agvar_1, \forall \agvar_2, \exists \agvar_1 \forall \agvar_2, \forall \agvar_2 \exists \agvar_1}$,
 where $\emptystring$ is the empty string. 
In each of the cases of the algorithms, $\pfix$ is assumed to be the \emph{longest} quantifier prefix that matches the input (sub)-formula.
 \begin{enumerate}
\item  The base case in Algorithm~\ref{algo:glob-mc} (line 3)
 of $\varphi$ being an atomic proposition $p$ simply returns the set of states, the labels of which contain $p$. 

\item  The boolean cases are straightforward.
 
\item 
In the case of Nexttime formula $\pfix \coop{\term_1, \term_2}\atlx \fob$, the algorithm first computes the state extension $\Qstates$ of the subformula $\fob$ with a recursive call, and then the controllable pre-image of $\Qstates$. 
The computation of the respective controllable pre-image is shown in Algorithm~\ref{algo:pre-img}. 
First, if any of $\term_1$ and $\term_2$ is not a variable
 that appears (i.e., is bound) in the quantifier prefix  $\pfix$, the assignment $\agass$ is applied to assign its value.
Then, for each state $\sts$, if the formula $\pfix \, \prf(\dmas, \sts, \term_1, \term_2, \Qstates)$ is true, the algorithm adds $\sts$ to the set of controllable states to be returned.

\item 
 Algorithms~\ref{algo:G-fixpoint} and~\ref{algo:U-fixpoint} compute the extension of 
 closed formulae of the type $\pfix  \coop{\term_1, \term_2} \chi$ 
 with temporal objective $\chi$ starting with $\atlg$ and $\atlu$ respectively. Their structure is similar to that for global model checking of such formulae in ATL (cf. e.g. the algorithm presented in \cite[ch.9]{TLCSbook}). They apply the iterative procedures of computing controllable pre-images that the fixpoint characterizations of the temporal operators $\atlg$ and $\atlu$ yield (ibid.). 
This is possible for quantified formulae as the quantifiers in formulae from \langb\ are propagated inside the temporal operators according to the respective fixpoint equivalences, proved in Theorem~\ref{lem:fixpoints1}.  
\end{enumerate}

\begin{theorem}
  Let $\dmas$ be a \oursys, $\varphi$ a $\langb$-formula and $\agass$ an assignment. Then
  \[
\stexten{\varphi}{\agass} = \algostyle{globalMC}(\dmas, \fo, \agass)
  \]
\end{theorem}

\begin{proof}
By induction on the structure of \langb formulae. The boolean cases are straightforward. For nexttime formulae $\pfix \coop{\term_1, \term_2} \atlx \psi$ the claim immediately follows from the correctness of 
Algorithm~\ref{algo:pre-img}, implied by the semantics of $\prf(\dmas, \term_1, \term_2, \pfix, \stexten{\psi}{})$. For formulae of the type $\pfix \coop{\term_1, \term_2} \atlg \psi$ and $\pfix \coop{\term_1, \term_2} \psi_1 \atlu \psi_2$, it follows from the correctness of  Algorithms~\ref{algo:G-fixpoint} and~\ref{algo:U-fixpoint}, justified  by Theorem~\ref{lem:fixpoints1}.
  \end{proof}

  For model checking of the full language \lang, Algorithm~\ref{algo:glob-mc}  is combined with function $\nff$, transforming constructively any \lang-formula $\varphi$ to $\varphi^{\mathsf{NF}}$ in \langb, equivalent in the finite to $\varphi$ by virtue of  Theorem \ref{thm:langb}.

  \begin{example}
  \label{ex:32}
    We illustrate Algorithm~\ref{algo:glob-mc} by sketching its application to the formula  
     $\psi = \coop{7, 4} \atlx (\forall \agvar_2 \exists \agvar_1 \coop{\agvar_1, \agvar_2} \atlg p)$ in the \oursys model $\dmas$ in Figure~\ref{fig:example}. We fix any  assignment $\agass$ (it does not play any role, since $\psi$ is closed). The outer formula is a $\atlx$ formula, thus line $18$ calls recursively the global model checking on the subformula in the temporal objective. Line $4$ of $\algostyle{G-fixpoint}$ initializes $\aset \gets \set{\sts_2, \sts_3, \sts_4}$, viz., states labeled with $p$ and $\asetbis \gets \States = \set{\sts_1, \ldots, \sts_6}$. Since $\asetbis \not \subseteq \aset$, we enter the while cycle computing the fixpoint. In the numbered list below, each item $i)$ correspond to the $i$-th iteration cycle. 
\begin{enumerate}
\item
  \begin{itemize}
  \item $\asetbis \gets \set{\sts_2,\! \sts_3,\! \sts_4}$;
    \item $\algostyle{preImg}(\dmas,\! \agvar_1,\! \agvar_2,\! \set{\sts_2,\! \sts_3,\! \sts_4},\! \agass,\! \forall \agvar_2 \exists \agvar_1)\!=\!\set{\sts_2,\! \sts_4,\! \sts_5}$;
\item $\aset \gets \set{\sts_2, \sts_4, \sts_5} \cap \set{\sts_2, \sts_3, \sts_4} = \set{\sts_2, \sts_4}$.
\end{itemize}
\item
\begin{itemize}
  \item $\asetbis \gets \set{\sts_2, \sts_4}$;
  \item $\algostyle{preImg}(\dmas,\! \agvar_1,\! \agvar_2,\! \set{\sts_2,\! \sts_4}, \agass,\! \forall \agvar_2 \exists \agvar_1) = \set{\sts_2,\! \sts_4,\! \sts_5}$;
    \item $\aset \gets \set{\sts_2, \sts_4, \sts_5} \cap \set{\sts_2, \sts_3, \sts_4} = \set{\sts_2, \sts_4}$.
    \end{itemize}
  Now $\asetbis \gets \aset$ then the fixpoint is reached.
\end{enumerate}
The set $\aset$ is then returned, so $\stexten{\forall \agvar_2 \exists \agvar_1 \coop{\agvar_1, \agvar_2} \atlg p}{} = \set{\sts_2, \sts_4}$. We now move to the outer next formula for which line 19 of \algostyle{globalMC} algorithm calls the \algostyle{preImg} procedure. For each $\sts \in \States$ the truth of formula $\prf(\dmas, \sts, 7, 4, \set{\sts_2, \sts_4})$ is called. The final result is $\stexten{\psi}{}=\set{\sts_4, \sts_5}$.
\end{example}

    
    \begin{example}
      \label{ex:33}
Consider $\varphi = $ 
$\coop{6, 3}\!\atlx\! \big(\exists \agvar_1 \coop{\agvar_1,\!  10}\! (\forall \agvar_2\exists \agvar_1 \coop{\agvar_1, \agvar_2}\! \atlg p) \! \atlu\! (\forall \agvar_2 \coop{0, \agvar_2} \! \atlg \! q)\!  \big)$.
We start by computing the extension of $\forall \agvar_2 \coop{0, \agvar_2} \atlg q$, following Algorithm~\ref{algo:G-fixpoint}. 
 
From lines $4-6$: \  
$\Qstates \gets \stexten{q}{} = \set{\sts_5, \sts_6}$; $\asetbis \gets \set{\sts_1, \ldots, \sts_6}$, and $\aset \gets \set{\sts_5, \sts_6}$. 

Since $\asetbis \not \subseteq \aset$, we enter the iteration cycle: 
\begin{enumerate}
\item
  \begin{itemize}
  \item $\asetbis \gets \set{\sts_5, \sts_6}$;
  \item $\algostyle{preImg}(\dmas, 0, \agvar_2, \set{\sts_5, \sts_6}, \agass, \forall \agvar_2) = \set{\sts_6}$
    \item $\aset \gets \set{\sts_6} \cap \set{\sts_5, \sts_6} = \set{\sts_6}$.
    \end{itemize}
  \item
    \begin{itemize}
    \item $\asetbis \gets \set{\sts_6}$;
    \item $\algostyle{preImg}(\dmas, 0, \agvar_2, \set{\sts_6}, \agass, \forall \agvar_2) = \set{\sts_6}$;
      \item $\aset \gets \set{\sts_6} \cap \set{\sts_5, \sts_6} = \set{\sts_6}$.
      \end{itemize}
      The fixpoint is reached and $\stexten{\forall \agvar_2 \coop{0, \agvar_2} \atlg q}{} = \set{\sts_6}$.
  \end{enumerate}
  From Example \ref{ex:32} we get $\stexten{\forall \agvar_2\exists \agvar_1 \coop{\agvar_1, \agvar_2} \atlg p}{}=\set{\sts_2, \sts_4}$. We then move to computing the extension of the until formula, following Algorithm~\ref{algo:U-fixpoint}. From lines $4-7$: \\ 
  $\Qstates_1 \gets \set{\sts_2, \sts_4}$; $\Qstates_2 \gets \set{\sts_6}$; $\asetbis \gets \emptyset$ and $\aset \gets \set{\sts_6}$. 
  
  Since $\aset \not \subseteq \asetbis$, we enter the iteration cycle: 
  \begin{enumerate}
  \item
    \begin{itemize}
    \item $\asetbis \gets \set{\sts_6}$;
    \item $\algostyle{preImg}(\dmas, \agvar_1, 10, \set{\sts_6}, \agass, \exists \agvar_1) = \set{\sts_4, \sts_6}$. 
    
    Indeed, from $\sts_4$, e.g., $40$ controllable agents performing $\act_1$ guarantee that guard $\guard_4$ is satisfied.
      \item $\aset \gets \set{\sts_6} \cup (\set{\sts_4, \sts_6} \cap \set{\sts_2, \sts_4} ) = \set{\sts_4, \sts_6}$.
      \end{itemize}
    \item
      \begin{itemize}
      \item $\asetbis \gets \set{\sts_4, \sts_6}$;
      \item $\algostyle{preImg}(\dmas, \agvar_1, 10, \set{\sts_4,\sts_6}, \agass, \exists \agvar_1) = \set{\sts_2, \sts_4, \sts_6}$;
        \item $\aset \gets \set{\sts_6} \cup (\set{\sts_2, \sts_4, \sts_6} \cap \set{\sts_2, \sts_4}) = \set{\sts_2, \sts_4, \sts_6}$.
        \end{itemize}
      \item
        \begin{itemize}
        \item $\asetbis \gets \set{\sts_2, \sts_4, \sts_6}$;
        \item $\algostyle{preImg}(\dmas,\! \agvar_1,\! 10,\! \set{\sts_2,\! \sts_4,\! \sts_6},\! \agass,\! \exists \agvar_1) = \set{\sts_2,\! \sts_4,\! \sts_5,\! \sts_6}$;
          \item $\aset \gets \set{\sts_6} \cup (\set{\sts_2, \sts_4, \sts_5, \sts_6} \cap \set{\sts_2, \sts_4}) = \set{\sts_2, \sts_4, \sts_6}$.
          \end{itemize}
          The fixpoint is reached. Thus: 
          \\
          $\stexten{\exists \agvar_1 \coop{\agvar_1,\!  10}\! (\forall \agvar_2\exists \agvar_1 \coop{\agvar_1, \agvar_2}\! \atlg p) \! \atlu\! (\forall \agvar_2 \coop{0, \agvar_2} \! \atlg \! q)}{} =  
          \set{\sts_2, \sts_4, \sts_6}$.
    \end{enumerate}
Lastly,  we call $\algostyle{preImg}(\dmas, 6, 3, \set{\sts_2, \sts_4,\sts_6}, \agass, \emptystring)$ to compute $\stexten{\varphi}{}  = \set{\sts_1, \sts_4, \sts_5, \sts_6}$. 
      \end{example}

%% file: complexity.tex
\section{Complexity estimates}
\label{sec:complexity}

As well-known from \cite{AHK-02}, the time complexity of model checking of ATL formulae is linear in both the size of the model\footnote{The simplified algorithm presented here works in quadratic time.} and the length of the formula. Note that in standard concurrent game models the number of agents is fixed and the transition relation is represented explicitly, by means of transitions from each state labelled with each action profile.
In \oursys models, however, the transitions are represented symbolically, in terms of the guards that determine them. An explicit representation would be infinite, in general. 
Thus, the question of how to measure the size of \hdmas models  arises. 
Given a \oursys $\dmas$, we consider the following parameters: the size $|\States|$ of the state space; the size $n$ of the action set 
$\actset$, and the size $|\transrel|$ of the symbolic transition guard function. The latter is defined as the sum of the length of all guards appearing in $\transrel$, where we assume a binary encoding of numbers.

Given a $\langb$ formula $\varphi$ and a \oursys \dmas,
the number of fixpoint computations in the global model checking algorithm is bounded by the length of $|\varphi|$.
Each computation executes the while cycle at most $|\States|$ times, and at each iteration, the function $\algostyle{preImage}$ is called. The pre-image algorithm cycles through all states again and invokes model checking of a \pra\ formula $\prf$ each time. In the worst case $|\prf| = |\transrel|$, as $\grd{\sts}{\Qstates}$ could be the disjunction of almost all guards in $\dmas$. The complexity of checking the truth of a \pra-formula depends not just on its size, but more precisely on the numbers of quantifier alternations and of quantified variables in any quantifier block (cf.~\cite{DBLP:journals/siglog/Haase18}). In our case, the maximum number of quantifier alternations is 4, while the number of variables in any quantifier block is at most $n+1$. By applying results from ~\cite{DBLP:conf/csl/Haase14} (cf. also~\cite{DBLP:journals/siglog/Haase18}), these yield a worst case complexity  $\Sigma_3^{\textsf{EXP}}$, or more precisely $\textrm{STA}(\ast, 2^{|\transrel|^{O(1)}} , 3)$ when the model is not fixed, or at least $n$ is unbounded, but it is down to $\textrm{STA}(\ast, {|\transrel|^{O(1)}}, 3)$ when $n$ is fixed.

\medskip
Thus, the number of variables and quantifier alternation depth in $\prf$-formulas crucially affect the complexity of model checking of $\langb$- formulae. We can distinguish the following cases of lower complexity bounds:

\begin{enumerate} 
\item When no quantifier patterns 
$\exists \agvar_1 \forall \agvar_2$
occur, the maximal alternation depth is 3, hence the complexity is reduced to $\textrm{STA}(\ast, 2^{|\transrel|^{O(1)}} , 2)$, respectively  $\textrm{STA}(\ast, {|\transrel|^{O(1)}} , 2)$. 

\item 
If no quantification $\forall \agvar_2$ is allowed, but the number of uncontrollable agents is a parameter, the maximal alternation depth is 2, hence the complexity is reduced to $\textrm{STA}(\ast, 2^{|\transrel|^{O(1)}} , 1)$, respectively  $\textrm{STA}(\ast, {|\transrel|^{O(1)}} , 1)$. 

\item In the case when the number of either controllable or uncontrollable agents is fixed or bounded, the resulting $\prf$-formulas become  either existential or universal (by replacing the quantifiers over the actions of the bounded set of agents with conjunctions, resp. disjunctions),
In these cases, the complexity drops to NP-complete if the number of actions is unbounded, resp. P-complete if that number is fixed or bounded.  
\end{enumerate}

%% file: concluding.tex
\section{Concluding remarks}
\label{sec:concluding}

{We have proposed and explored a new, generic framework for modelling, formal specification and verification of dynamic multi-agent systems, where agents can freely join and leave during the evolution of the system. We consider indistinguishable agents and therefore the system evolution is affected only by the number of agents performing actions. As neither of the currently available logics are well-suited for expressing properties of such dynamic models, we have devised a variation of the alternating time temporal logic ATL to specify strategic abilities of coalitions of controllable versus non-controllable of agents.}

The framework and results presented here are amenable to various extensions, e.g. allowing any \pra-formulae as guards in \hdmas models; allowing more expressive languages, e.g. with arbitrary LTL or parity objectives, with somewhat more liberal quantification patterns in \lang (i.e., formulae of the type $\forall y \coop{y,y} \atlx \fo$ and $\exists y \coop{y,y} \atlx \fo$ can be added easily), adding several super-agents with controllable sets of agents, etc.  The main technical challenge for some of these extensions would be to lift or extend the model checking procedure for them. Still, in particular, extending the present framework to include any finite number of different agent ``types'', with each type having a different protocol, is rather straightforward, as follows. 
Let us fix a set of agent types $\set{T_1, \ldots, T_m}$. Now each agent belong to one  specific type. Definition~4 will then have $d_1, \ldots, d_m$ action availability functions, one for each type, so that agents belonging to the same type have the same set of available actions in each system state, but agents belonging to different types might have different available actions. Lastly the logic will now involve $m$ 
variables for the controllable agents of each type, and $m$ other variables for the non-controllable ones of each type.     
The same restrictions on the use of these variables will apply in this extended logic and the notion of normal form, the technical results related to it, and the model checking algorithm for formulae in normal form, extend as expected to the multi-type case. 

Of the numerous possible applications we only mention a natural link with the Colonel Blotto games \cite{borel1953theory}, \cite{Roberson2006}, where two players simultaneously distribute military force units across $n$ battlefields, and in each battlefield  the player (if any) that has allocated the higher number of units wins. 
As suggested by our fortress example, our framework can be readily applied to model and solve algorithmically multi-player and multiple-round extensions of Colonel Blotto games, which we leave to future work. More generally, dynamic resource allocation games \cite{AvniHenzingerKupferman16} as well as verification of parameterised fault-tolerance in multi-agent systems \cite{DBLP:conf/ijcai/KouvarosL17} seem naturally amenable to applications of the present work.